\documentclass[a4paper, 11pt]{article}
\usepackage[utf8]{inputenc}
\usepackage[english]{babel}
\usepackage[margin=1in]{geometry}

\usepackage{graphicx} 
\usepackage{amssymb}
\usepackage{amsmath}
\usepackage{bbm}
\usepackage{float}
\usepackage{subcaption}
\usepackage{comment}
\usepackage{bm}
\usepackage{amsthm}

\RequirePackage{natbib}
\usepackage[colorlinks,allcolors=blue,linktocpage=true]{hyperref}

\newcommand{\PP}{\mathbb{P}}
\newcommand{\QQ}{\mathbb{Q}}
\newcommand{\E}{\mathbb{E}}
\newcommand{\Var}{\text{Var}}
\newcommand{\norm}[1]{\left\lVert #1 \right\rVert}
\newcommand{\indep}{\perp \!\!\! \perp}
\newcommand{\overbar}[1]{\mkern 1.5mu\overline{\mkern-1.5mu#1\mkern-1.5mu}\mkern 1.5mu}

\newtheorem{proposition}{Proposition}
\newtheorem{theorem}{Theorem}
\newtheorem{corollary}{Corollary}
\newtheorem{lemma}{Lemma}
\newtheorem{remark}{Remark}

\title{Population Size Estimation with Many Lists and Heterogeneity: A Conditional Log-Linear Model Among the Unobserved}
\author{Mateo Dulce Rubio \and Edward Kennedy}
\date{Carnegie Mellon University}

\begin{document}

\maketitle

\begin{abstract}
\noindent

   We contribute a general and flexible framework to estimate the size of a closed population in the presence of $K$ capture-recapture lists and heterogeneous capture probabilities. Our novel identifying strategy leverages the fact that it is sufficient for identification that a subset of the $K$ lists are not arbitrarily dependent \textit{within the subset of the population unobserved by the remaining lists}, conditional on covariates. This identification approach is interpretable and actionable, interpolating between the two predominant approaches in the literature as special cases: (conditional) independence across lists and log-linear models with no highest-order interaction. We derive nonparametric doubly-robust estimators for the resulting identification expression that are nearly optimal and approximately normal for any finite sample size, even when the heterogeneous capture probabilities are estimated nonparametrically using machine learning methods. Additionally, we devise a sensitivity analysis to show how deviations from the identification assumptions affect the resulting population size estimates, allowing for the integration of domain-specific knowledge into the identification and estimation processes more transparently. We empirically demonstrate the advantages of our method using both synthetic data and real data from the Peruvian internal armed conflict to estimate the number of casualties. The proposed methodology addresses recent critiques of capture-recapture models by allowing for a weaker and more interpretable identifying assumption and accommodating complex heterogeneous capture probabilities depending on high-dimensional or continuous covariates.

\end{abstract}
\section{Introduction}

Estimating the size of a population is a critical challenge across various domains, such as biology for estimating animal populations and conservation efforts \citep{schwarz1999estimating, burgar2018estimating}, public health for estimating disease prevalence \citep{hook1995capture, bohning2020estimating}, and human rights for estimating the number of victims of specific crimes \citep{manrique2013multiple, silverman2020multiple}. In many high-stakes applications, we observe \textit{capture-recapture} data, where each capture episode involves a subsample from the population and records it in a \textit{list}. In such data, the same unit can be observed multiple times, while some units are never captured. Intuitively, the intersection between capture lists can be exploited to estimate the total population size, as smaller populations tend to have units appearing in multiple lists more frequently than larger populations.

Various approaches have been used for estimating the size of a population from multiple-list capture-recapture data. Among these methodologies, one popular approach assumes that capture episodes are independent and that all units in the target population have the same probability of being observed. Under these assumptions, the observed data can be used to estimate the number of unobserved units by extrapolating the estimated probabilities of each observed capture profile \citep{bishop2007discrete, otis1978statistical}. However, the assumptions of independence and homogeneity can be implausible in real-world applications, which limit the applicability and robustness of the resulting estimators. For instance, in settings where demographic or biological attributes make some
units more or less easy to capture, assuming homogeneous capture probabilities may bias the estimated population size in arbitrary ways. Therefore, there is a pressing need for innovative statistical methodologies in the face of recent and challenging applications, such as estimating the population size of HIV patients \citep{wesson2024evaluating}, casualties during armed conflicts \citep{manrique2021capture}, or victims of modern slavery \citep{binette2022reliability}, particularly in the presence of heterogeneous capture probabilities and complex dependencies between capture episodes.

To move away from the assumptions of homogeneity and independence, more general approaches have been proposed, such as Bayesian latent class models \citep{manrique2016bayesian}, and methods leveraging semiparametric efficiency theory \citep{you2021estimation, das2023doubly}. However, these methods rely on identification assumptions that can be difficult to understand, making it challenging to argue for their validity in practical applications. Furthermore, they sometimes conflate issues of identification and estimation \citep{aleshin2024central}, or use only part of the available information collected during the capture episodes.  Despite these advances, there remains a need for a general framework that allows for an interpretable, useful, and plausible identifying assumption and accommodates heterogeneous capture probabilities in a flexible manner, even when this heterogeneity depends on high-dimensional or continuous attributes. Developing such a general and flexible framework for population size estimation is the ultimate goal of our work.

\subsection{Contributions}

Concretely, we contribute to the population size estimation literature in four key ways: 

\begin{itemize}
    \item We introduce a novel conditional log-linear model that incorporates heterogeneous capture probabilities according to a set of available covariates, allowing for high-dimensional and continuous attributes.
    
    \item We propose a new identification strategy by noting that it is sufficient for identification that a subset of the capture-recapture lists are not arbitrarily dependent within the subset of the population unobserved by the remaining lists, conditional on covariates. This strategy is both interpretable and actionable, generalizing the two predominant approaches in the literature: conditional independence between capture episodes and traditional log-linear models with no highest-order interaction.
    
    \item We derive nonparametric doubly-robust estimators for the capture probability and the total population size. The proposed one-step estimators are shown to be nearly minimax optimal and approximately normal for any finite sample size, providing robust and reliable estimates even under flexible estimation of the heterogeneous capture probabilities using machine learning models.
    
    \item Finally, we introduce a sensitivity analysis framework that reveals how violations of the identification assumptions impact the population size estimates. This allows researchers to incorporate domain-specific knowledge into the estimation process more transparently and rigorously.
\end{itemize}

To validate our methodology, we conduct empirical evaluations using both synthetic data and real data from the Peruvian internal armed conflict. Our results highlight the advantages of our approach in accurately estimating population size in real-world applications, allowing for a weaker and more interpretable identifying assumption and accommodating complex heterogeneous capture probabilities.

\paragraph{Paper outline} Section 2 describes the setup of the capture-recapture problem and the predominant identification strategies used in previous work. Section 3 presents our conditional log-linear model for heterogeneous capture probabilities and introduces our novel identification strategy to express the capture probability using the observed data. In Section 4, we derive a one-step estimator using semiparametric efficiency theory and show that is nearly optimal and approximately normal for any finite sample size, even under flexible estimation of the heterogeneous capture probabilities. In Section 5, we devise a sensitivity analysis and construct analogous one-step estimators that bound the capture probability, while Section 6 extends our results to the total population size. Section 7 illustrates the advantages of the proposed methods, over a natural plug-in estimator, in simulated and real data from the Peruvian internal armed conflict during 1980-2000. Finally, Section 8 concludes and discusses the broad applicability of the proposed framework for population size estimation across various domains.
\section{Problem Setup}
\label{preliminaries}

Consider a closed population of unknown size $n$ that we aim to estimate. Suppose that we have access to $K$ partial lists, each capturing different subsets of this population. Then, each uniquely identified subject $i \in \{1, \dots, n\}$ has an associated \textit{capture profile} vector:
$$ Y_i = (Y_{i1}, \dots, Y_{iK}), $$
where $Y_{ik} \in \{0,1\}$ indicates whether unit $i$ appears on list $k \in \{1, \dots, K\}$. In addition, let us assume that we observe a set of covariates $X_i \in \mathbb{R}^d$ for each subject $i$. Crucially, note that we only observe the $N \leq n$ units for whom $Y_i \neq 0$, where $0$ denotes the $K$-dimensional zero vector. That is, we only observe those that appear in at least one of the $K$ lists. Our goal is then to use these $N$ \textit{capture-recapture} observations to estimate the total population size $n$.

In this setup, the total number of unique observed units, $N = \sum_{i=1}^{n} \mathbbm{1}(Y_i \neq 0)$, is a random variable that follows a binomial distribution: $N \sim \text{Bin}(n, \psi)$, where $\psi = \PP(Y \neq 0)$ is the \textit{capture probability}. Treating the population size $n$ as a fixed parameter satisfying $\E[N] = \psi n$, we can derive a straightforward estimator for $n$ based on the observed sample size $N$:
\begin{equation}
\label{eq:estimator_n}
    \widehat{n} = \frac{N}{\widehat{\psi}}.
\end{equation}
Therefore, estimating the population size $n$ reduces to estimating the inverse capture probability $\psi^{-1}$ from the observed data.

However, identifying and estimating $\psi = 1 - \PP(Y = 0)$ is notably challenging precisely because we do not observe any of those units with $Y_i = 0$. Intuitively, we aim to estimate $\PP(Y = 0)$ by extrapolating the information from those with $Y \neq 0$ \citep{manrique2021capture}, which may be arbitrarily different populations. This is a hopeless task unless we introduce additional \textit{identifying assumptions} that allow us to use the information from the observed distribution to identify and estimate the population parameter $\psi$ using the observed data.\footnote{We refer the reader to \cite{aleshin2024central} for a detailed discussion on the role of the identification assumption in population size estimation.}

Formally, we assume that the covariate and conditional list membership distributions are the same for all units, and that these probabilities do not depend on other units' variables. This implies that the random vectors $(X_i, Y_i)$ are \textit{iid} with respect to some (unknown) underlying distribution $\PP$. Crucially, the capture-recapture setting suffers from biased sampling in the sense that the observed data $\{(X_i, Y_i)\}_{i=1}^{N}$ are actually iid draws from the conditional distribution
$$\QQ(Y=y, X=x) \equiv \PP(Y = y, X=x \mid Y\neq 0).$$
That is, the observed distribution $\QQ$ captures the joint distribution of $(X,Y)$ conditional on them being observed, i.e., given $Y\neq 0$.

To deal with such biased sampling, let $q_y(x) = \QQ(Y = y \mid X=x)$ denote the conditional probability that the list membership vector $Y = (Y_1, \dots, Y_K)$ equals the specific capture profile $y = (y_1, \dots, y_k) \in \{0,1\}^{K}$, given covariates $X=x$. These measure the capture heterogeneity in the target population, and we will refer to the set $\{q_y(X) = \QQ(Y = y \mid X)\}_{y\neq 0}$ as \textit{$q$-probabilities} \citep{das2023doubly}. Then, note that for any $y\neq 0$
\begin{equation}
\label{eq:p_qgamma}
    \PP(Y=y \mid X=x) = \PP(Y=y \mid X=x, Y\neq 0)\PP(Y\neq 0 \mid  X=x) = q_y(x)\gamma(x),
\end{equation}
where $\gamma(x) = \PP(Y\neq 0 \mid X=x)$ is the \textit{conditional capture probability} for a subject with covariates $X=x$. Finally, note that we can express $\psi^{-1} = \PP(Y\neq 0)^{-1}$ as the harmonic mean of $\gamma(X)$ under distribution $\QQ$ \citep{das2023doubly, johndrow2019low}:
\begin{equation}
    \frac{1}{\psi} = \E_{\QQ}\left[ \frac{1}{\gamma(X)}\right] = \int \frac{1}{\gamma(x)}\QQ(X=x)dx,    
    \label{eq:harmonic_mean}
\end{equation}
where $\QQ(X=x) = \PP(X=x \mid Y\neq 0)$. Therefore, if we can identify the conditional capture probability $\gamma(X)$ from the available data, we can use the previous identity \eqref{eq:harmonic_mean} to identify the inverse capture probability $\psi^{-1}$ using the observed data from $\QQ$.

\subsection{Identification strategies in capture-recapture settings}

In this section we summarize two popular identification strategies used in the capture-recapture literature to identify the capture probability and the total population size: conditional independence given covariates, and no highest-order interaction in a log-linear model with many lists. Intuitively, since the capture-recapture setting suffers from biased sampling, it is necessary to introduce additional assumptions that enable us to extrapolate the information from distribution $\QQ$, observed for those that appear at least on one list, to the distribution $\PP$, in particular for those who are unobserved. Crucially, if lists are irreparably dependent we will not be able to identify the target population size \citep{aleshin2024central}.

For instance, in the $K=2$ list case without covariates, assuming that the probability of appearing in one list is independent of appearing on the other list allows for the total population size to be estimated directly using the Lincoln-Petersen estimator \citep{lincoln1930calculating, petersen1896yearly}. This assumption is untestable and can be implausible in practical settings, but is relaxed by measuring covariates that potentially affect the probability of being captured. Then, using the identification assumption of \textit{conditional} independence between lists, the population size $n$ can be correctly identified and estimated from observed data \citep{rivest2014capture, das2023doubly}. Nonetheless, this is still considered a strict identification assumption, implicitly assuming that all relevant variables can be measured, and typically only allowing for the use of two lists.

An alternative common approach when $K$ capture-recapture lists are available is to assume an underlying log-linear model. These models assume that the probability of observing each capture profile vector $y$ satisfies \citep{fienberg1972multiple}
\begin{equation}
\label{eq:normal_loglin}
    \log(\PP(Y=y)) = \alpha_0 + \sum_{y' \neq 0} \alpha_{y'} \prod_{k:~ y'_k = 1}y_k,
\end{equation}
for some parameters $\alpha_y$, where the sum is over all possible non-zero capture profiles $y'\in \{0,1\}^{K}$. In this setting, log-linear models are unidentifiable given that there are $2^K$ parameters $\alpha_y$ to be estimated from $2^K - 1$ observed capture profiles; all but $y=0$. Therefore, the standard approach is to assume \textit{no highest-order interaction}, $\alpha_{1} = 0$, which assumes that the probability $\PP(Y = 1)$ can be expressed using only lower-level interaction terms captured by the $\alpha_y$ parameters, for $y\neq 1$ \citep{fienberg1972multiple, bishop2007discrete}.  We provide a more detailed interpretation of the no highest-order interaction assumption in terms of conditional odds ratios in the Appendix.  However, these models require overlap between all combinations of lists and typically assume homogeneous capture probabilities, meaning that all units have the same probability of being observed, which can be considered strong assumptions. For instance, in settings where demographic or biological attributes make some units more or less easy to capture, assuming homogeneous capture probabilities may bias the estimated population size in complicated ways \citep{manrique2021capture}. 

In summary, there are two predominant approaches for identifying the total population size from capture-recapture data: (i) measuring all possible features and invoking a conditional independence assumption \citep{das2023doubly}, or (ii) collecting multiple lists and using a log-linear model with no highest-order interaction term \citep{you2021estimation}. In general, for any number of $K$ lists, assuming that two lists are conditionally independent given the other lists and covariates is a stricter identification assumption than the no highest-order interaction in a conditional log-linear model. The latter identifying assumption accounts for more complex capture-recapture dependencies across lists, but is harder to interpret for a number of lists larger than $K>3$. In addition, standard log-linear models requires positivity for all capture profiles and cannot account for heterogeneous capture probabilities. 

In the next section, we demonstrate how the two predominant strategies can be integrated by using conditional log-linear models combined with a weaker and more interpretable identification assumption. The proposed identification strategy generalizes the two predominant approaches noting that it is sufficient for identification that a \textit{subset} of the $K$ lists has no highest-order interaction, conditional on not being observed by the remaining lists and covariates. For instance, the assumption 
$$Y_1 \indep Y_2 \mid X, Y_3 = 0, \dots, Y_K=0,$$
is sufficient to identify the capture probability $\psi$. Note that the no highest-order interaction assumption and conditioning only on not being observed by the remaining lists do not impose any restrictions on the observed distribution $\QQ$, allowing for flexible \textit{nonparametric} data-driven modeling of $\psi$. 

\section{Conditional Log-Linear Model Among the Unobserved}

In this section, we give a new approach that incorporates heterogeneous capture probabilities into a conditional log-linear model where list membership varies across covariates $X$. We introduce a new identification approach by assuming there is a subset of $J \leq K$ lists that are not arbitrarily dependent given covariates, \textit{in the subset of the population unobserved by the remaining $K-J$ lists}. We show that this identifying assumption is sufficient to express the target capture probability using the observed data, and that it generalizes the two most common identification strategies used in previous literature. In short, we derive an identifying expression for $\gamma(X)$ from the proposed conditional log-linear model among the unobserved, which we use to identify the capture probability using \eqref{eq:harmonic_mean} and propose efficient estimators for $\psi^{-1}$. Finally, we discuss how to interpret and use the identifying assumption introduced in our work.

In more detail, we first fix a subset $Y_1, \dots, Y_J$ of $J$ lists ($J\leq K$).\footnote{For simplicity, we will index the subset of $J$ lists as $Y_1, \dots, Y_J$. But we do not assume these are captured before lists $Y_{J+1}, \dots, Y_{K}$ or in any particular order.} Then, we use a conditional log-linear model to express the (log) conditional probability that $(Y_1, \dots, Y_J) = (y_1, \dots, y_J)$, conditional on covariates $X$ and on not being present in the other $K-J$ lists. That is, for those unobserved by lists $Y_{J+1}, \dots, Y_K$. Mathematically,
\begin{equation}
\label{eq:log_lineal}
    \log\left[\PP\left\{(Y_1, \dots, Y_J) = (y_1, \dots, y_J) \mid X, (Y_{J+1},\dots,Y_K)=0\right\}\right] = \alpha_0(X) + \sum_{(y'_1, \dots, y'_J) \neq 0} \alpha_{y'}(X) \prod_{j:~ y'_j = 1}y_j,
\end{equation}
where the $\alpha_y(X)$  are arbitrary. Notably, treating $\alpha_y(X)$ as parameters depending on the covariates $X$ allows us to account for heterogeneous capture probabilities of the subjects in the target population, including high-dimensional and continuous attributes.  

\begin{remark}
    To ensure that our log-linear model is well-defined, formally it is required that $\PP(X, (Y_{J+1},\dots,Y_K)=0) > 0$.  This holds by design in the capture-recapture setting because we assume that there is still a portion of the population not captured by the available lists, and in particular by the $J+1, \dots, K$ lists. That is, $\PP((Y_{J+1},\dots,Y_K)=0 \mid X) > 0$. Otherwise, we could use the observed records directly as the total population size.    
\end{remark}


\subsection{Identification in a conditional log-linear model}

To identify the conditional capture probability we follow the standard identification strategy in log-linear models and assume that the highest-order interaction parameter in \eqref{eq:log_lineal} is zero, but only conditionally on covariates $X$ \citep{fienberg1972multiple}. Recall that the highest-order coefficient $\alpha_1(X)$ in a log-linear model captures higher order interactions between the capture episodes, and it is used to encode certain lack of dependence in the population distribution $\PP$. This assumption of no arbitrary dependence is necessary because, as explained above, otherwise the observed data distribution $\QQ$ would be uninformative about the target estimand $\psi = \PP(Y \neq 0)$. 

Building on the saturated log-linear model \eqref{eq:log_lineal}, we can express the highest-order interaction coefficient, $\alpha_1(X)$, in terms of the $q$-probabilities and the conditional capture probability $\gamma(X)$, using the relationship between the true distribution $\PP$ and the observed distribution $\QQ$ in \eqref{eq:p_qgamma}:
\begin{equation}
     \alpha_1(X) = \sum_{(y_1, \dots, y_J) \neq 0} (-1)^{J + |y|}\log(q_y(X)\gamma(X)) + (-1)^{J}\log(1-\gamma(X)),
\end{equation}
where $|y| = \sum_{j=1}^{J} y_j$ is the number of positive lists in the capture profile $y = (y_1, \dots, y_J)$ ($L^{1}-$norm) \citep{bishop2007discrete, you2021estimation, binette2022reliability}. Then, assuming $\alpha_1(X) = 0$, 
we obtain the following point-identification result for the capture probability $\psi$. We defer all the proofs in our work to the Appendix.

\begin{proposition} 
\label{prop:point_psi} \textbf{(Identification Result)} 
Assume no highest-order interaction in a conditional log-linear model for $J$ lists conditional on covariates and on not being captured by the remaining $K-J$ lists. Assume positivity $q_y(X) > 0$, \textit{with probability 1}. Then, we can identify the inverse capture probability $\psi^{-1}$ in terms of the observed data by
\begin{equation}
\label{eq:point_psi}
    \frac{1}{\psi} = \E_{\QQ}\left[\frac{1}{\gamma(X)}\right] = \E_{\QQ}\left[1 + \exp\left(\sum_{(y_1, \dots, y_J) \neq 0} (-1)^{|y| + 1}\log(q_y(X)) \right)\right],
\end{equation}
where the $q$-probabilities $q_y(X) = \QQ((Y_1,\dots, Y_{J}, Y_{J+1},\dots,Y_K) =(y_1,\dots,  y_{J}, 0,\dots, 0) \mid X)$ are computed for the complete $K$-dimensional capture profiles $y$ of the form $y=(y_1,\dots, y_J, 0, \dots, 0)$.
\end{proposition}

Our identification strategy is weaker than assuming a log-linear model for all $K$ lists with no highest-order interaction in two ways. First, when $J=K$ we recover the standard log-linear model \eqref{eq:normal_loglin}, conditional on covariates $X$, which assumes a lack of dependence between all available $K$ lists in the target population. Notably, our proposed approach only requires that $J\leq K$ lists are not arbitrarily dependent \textit{in the subset of the population unobserved by the rest of the $K-J$ lists}, conditional on covariates. Moreover, our identifying assumption captures the right intuition that identification is easier when (i) there are more relevant covariates $X$ available, or (ii) there is a larger number $K$ of lists that already capture a portion of the target population. 

Secondly, note that we only require positivity of the $q$-probabilities for the $2^{J}-1$ capture profiles of the form $y=(y_1,\dots, y_J, 0, \dots, 0) \neq 0$. On the other hand, using a conditional log-linear model on all of the $K$ lists needs the $2^{K}-1$ $q$-probabilities for all possible capture profiles $y=(y_1,\dots, y_K) \neq 0$ to be bounded away from zero with probability one. In the presence of many lists, this is a strong assumption due to potential lack of overlap between all possible combinations of lists. For instance, if $J=3$ and $K=5$, our approach assumes positivity for seven $q$-probabilities, while traditional log-linear models need positivity to hold for $31$ of them. This has been studied in the capture-recapture literature using data models with \textit{structural zeros} \citep{bishop2007discrete}. However, such methods typically introduce additional parametric assumptions mixing the identification and the estimation of the total population size. Our method overcomes this problem nonparametrically and does not require overlap between all $K$ lists, but only between the $J$ lists present in the conditional log-linear model.

Another special case arises when $J=2$. In this scenario, even in the presence of a large $K$, our identification strategy reduces to assuming conditional independence between two lists in the population unobserved by $Y_3, \dots, Y_K$, given covariates. The obtained identification expression resembles the result in \cite{das2023doubly}, which can be seen as a nonparametric conditional Lincoln-Peterson estimator for heterogeneous populations. Crucially, unlike the \cite{das2023doubly} estimator, our method uses the information from all $K$ lists in both the identification and estimation procedures. We summarize this special case in the following Corollary as we believe it can be useful in many settings where there are two larger lists that are believed to be (conditionally) independent.

\begin{corollary}
\label{cor:l2_condind}
    When $J=2$, our identification assumption is equivalent to assuming that two lists are conditionally independent given covariates for those not captured by $Y_3, \dots, Y_K$:
\begin{equation}
    Y_1 \indep Y_2 \mid X, Y_3 = 0, \dots, Y_K = 0.
\end{equation}
In this case, identification of the capture probability $\psi$ simplifies to
\begin{equation}
    \frac{1}{\psi} = \E_{\QQ}\left[\frac{1}{\gamma(X)}\right] = \E_{\QQ}\left[\frac{q_{10}(X)q_{01}(X)}{q_{11}(X)} + 1\right],
\end{equation}
where $q_{y_1y_2}(X) = \QQ(Y_1=y_1, Y_2=y_2,Y_3=0, \dots, Y_K=0 \mid X)$.
\end{corollary}

\begin{remark}
    The conditional log-linear model and the no highest-order interaction assumption are used to identify the target parameter $\psi^{-1}$ by assuming a lack of dependence of some of the capture episodes in the population distribution $\PP$. Nonetheless, this does not impose any restrictions on the \textit{observed} distribution $\QQ$, such that our model is nonparametric in that we do not constrain the nuisance functions  $q_y(x) = \QQ(Y=y \mid X=x)$ needed to identify $\psi$. Therefore, we can leverage semiparametric efficiency theory to propose efficient estimators without introducing additional potentially restrictive assumptions, allowing for flexible estimation of the nuisance heterogeneous $q$-probabilities. 
\end{remark}

To summarize, we have introduced a novel identification strategy assuming that a subset of $J\leq K$ lists are not arbitrarily dependent within the population not captured by the rest of the lists, conditional on covariates. This identification assumption generalizes the two predominant approaches in the literature: conditional independence between capture episodes and traditional log-linear models with no highest-order interaction between $K$ lists. Moreover, it is interpretable and actionable in the sense that to make identification of the capture probability easier, an analyst can either measure more relevant covariates or collect additional lists. In the next section, we use the obtained identifying expression to derive the efficiency bound to estimate $\psi^{-1}$ and propose a ``doubly-robust'' estimator that closely attains such bound and is approximately normal in finite samples.

\section{Efficient Nonparametric Estimation}

In this section, we build on the identifying expression from Proposition \ref{prop:point_psi} arising from the identification assumption of no highest-order interaction in a conditional log-linear model among the unobserved. We explicitly derive the Efficient Influence Function for \eqref{eq:point_psi} and the corresponding efficiency bound for estimating $\psi^{-1} = \E_{\QQ}[\gamma(X)^{-1}]$. We then introduce a novel one-step estimator for the capture probability and show that it is nearly optimal in a ``doubly-robust'' fashion and approximately Gaussian for any sample size $N$.

\subsection{Efficient influence function \& efficiency bound}

Before introducing our proposed estimator for the inverse capture probability, we first benchmark the estimation problem using the semiparametric efficiency bound. Importantly, the efficiency bound gives the best possible (scaled) Mean Squared Error that can be achieved for estimating $\psi^{-1}$ without introducing additional, potentially restrictive, parametric assumptions on the nuisance functions $q_y(X) = \QQ(Y=y \mid X)$. Therefore, in the absence of a correct parametric model for the $q$-probabilities, having an estimator that attains such efficiency bound is critical to allow for the flexible estimation of the nuisance functions capturing heterogeneous capture probabilities in a data-driven way. 

Formally, the efficiency bound for estimating the inverse capture probability can be computed as the variance of the \textit{Efficient Influence Function} (EIF) for $\psi^{-1}$. Therefore, we first use standard machinery from semiparametric efficiency theory to derive the corresponding EIF for the inverse capture probability from the identifying expression in 
equation \eqref{eq:point_psi}  (cf.  \cite{ehk_semiparametric, semiparametric, you2021estimation}).

\begin{lemma}
\label{lemma:if_psi}
     Assume (i) no highest-order interaction in a conditional log-linear model for $J$ lists, conditional on covariates and on being unobserved by the remaining $K-J$ lists, and (ii) positivity for the relevant $q$-probabilities. Then, the Efficient Influence Function for $\psi^{-1} = \E_{\QQ}[\gamma(X)^{-1}]$ is given by
\begin{equation}
\label{eq:if_psi}
    \phi(X,Y) = \left(\frac{1}{\gamma(X)} - 1\right)\sum_{(y_1, \dots, y_J) \neq 0} (-1)^{|y|+1} \frac{\mathbbm{1}(Y=y)}{q_y(X)} + 1 - \frac{1}{\psi},
\end{equation}
where $\gamma(X)^{-1} = 1 + \exp\left(\sum_{(y_1, \dots, y_J) \neq 0} (-1)^{|y| + 1}\log(q_y(X)) \right)$ is the identifying expression for the conditional capture probability. As in Proposition \ref{prop:point_psi}, the $q$-probabilities are computed for the $K$-dimensional capture profiles $y$ of the form $y=(y_1,\dots, y_J, 0, \dots, 0)$.
\end{lemma}

In the following Theorem, we use this result to compute the semiparametric efficiency bound to estimate the target estimand of the inverse capture probability, $\psi^{-1} = \E_{\QQ}[\gamma(X)^{-1}]$. Crucially, this benchmarks our estimation problem because it gives the best possible scaled MSE that can be achieved for estimating $\psi^{-1}$ without introducing additional parametric assumptions.

\begin{theorem}
\label{thm:eff_bound}
Under the same assumptions as in Lemma \ref{lemma:if_psi}, the semiparametric efficiency bound for estimating $\psi^{-1} = \E_{\QQ}[\gamma(X)^{-1}]$ takes the form
\begin{equation}
\label{eq:eff_bound}
    \Var(\phi(X,Y)) = \Var\left(\frac{1}{\gamma(X)}\right) + \E\left[\left(\frac{1}{\gamma(X)} - 1\right)^2\sum_{(y_1, \dots, y_J) \neq 0} \frac{1}{q_y(X)} - 1\right],
\end{equation}
where $\gamma(X)^{-1} = 1 + \exp\left(\sum_{(y_1, \dots, y_J) \neq 0} (-1)^{|y| + 1}\log(q_y(X)) \right)$ and the $q$-probabilities are computed for the $K$-dimensional capture profiles $y=(y_1,\dots, y_J, 0, \dots, 0)$.
\end{theorem}

Notably, the efficiency bound benchmarking our estimation problem is composed by two terms. The first term captures the heterogeneity in the conditional capture probabilities $\gamma(X)$, such that if these are more homogeneous we can estimate the marginal capture probability better. The second term involves the magnitude of these conditional probabilities and the magnitude of the individual $q$-probabilities for the capture profiles $y=(y_1,\dots, y_J, 0, \dots, 0)$. Intuitively, as those unobserved by $Y_{J+1}, \dots, Y_{K}$ are more easily observed under any combination of the lists $Y_1,\dots,Y_J$, we can estimate $\psi^{-1}$ more efficiently. 

Note that for a standard log-linear model on $K$ lists, the resulting efficiency bound would include the $q$-probabilities for all possible capture profiles $y = (y_1, \dots, y_K)$. Therefore, when there are rare capture profiles (i.e., $q_y(X)$ is small for some $y$'s), the efficiency bound is larger and  estimation is more difficult. This shows how our approach not only allows for weaker identification assumptions, but also makes estimation easier. This is illustrated empirically in our data analysis section using real data from the Peruvian internal conflict.

\subsection{Nonparametric estimation}

We now build over the previous results to propose a \textit{one-step} double machine-learning estimator that allows for flexible estimation of the $q$-probabilities, for example using machine learning models. We provide finite-sample guarantees for our one-step estimator showing that it is close to an empirical average of the EIF, which attains the semiparametric efficiency bound derived in \eqref{eq:eff_bound}, with high probability (nearly optimal). Moreover, we show that it is approximately normal allowing for standard inference procedures. We complement our analysis by comparing the proposed estimator to a natural, but typically sub-optimal plug-in estimator for $\psi^{-1}$. 

\subsubsection{Plug-in estimator}

To motivate the proposed one-step estimator we first introduce a simple plug-in estimator for $\psi^{-1}$. Concretely, from the identification expression in equation \eqref{eq:point_psi},
a straightforward \textit{plug-in estimator} to estimate $\psi^{-1} = \E_{\QQ}[\gamma(X)^{-1}]$ from observed data $\{(X_i, Y_i)\}_{i=1}^N$ is computed by replacing the conditional capture probability $\gamma(X)$ by an estimate $\widehat{\gamma}(X)$, and by estimating the expectation $\E_{\QQ}$ with the empirical average $\QQ_N$. Mathematically:
\begin{equation}
    \widehat{\frac{1}{\psi_{pi}}} = \QQ_N\left(\frac{1}{\widehat{\gamma}(X_i)}\right) = \frac{1}{N} \sum_{i=1}^N\left(1 + \exp\left(\sum_{(y_1, \dots, y_J) \neq 0} (-1)^{|y| + 1}\log(\widehat{q}_y(X_i)) \right)\right),
\end{equation}
where $\widehat{q}_y(X_i)$ are estimates of the $q$-probabilities, $q_y(X_i) = \QQ(Y=y \mid X=X_i)$, for each capture profile $y = (y_1, \dots, y_J, 0, \dots, 0) \neq 0$. 

Note that in constructing the plug-in estimator for the (inverse) capture probability, one estimates conditional capture probabilities for each unit, $\widehat{\gamma}(X_i)$, and their empirical harmonic mean (i.e., with respect to $\QQ_N$) is computed to obtain $\widehat{\psi}_{pi}^{-1}$. Crucially, a correct parametric model for the $q$-probabilities is often not available and a natural alternative is to estimate them using flexible methods, such as machine learning models, especially when the heterogeneous capture probabilities depend on high-dimensional or continuous covariates. However, such data-adaptive methods can cause $\widehat{\psi}_{pi}^{-1}$ to exhibit slower than $\sqrt{N}$-convergence rates and to have an unknown limiting distribution. The first limitation makes the plug-in estimator inefficient, while the second prohibits us from making any inference about $\widehat{\psi}_{pi}^{-1}$. 

\subsubsection{One-step estimator}

To overcome these limitations of the plug-in estimator in semi- and non-parametric settings, we leverage semiparametric efficiency theory and propose the following one-step estimator using the Efficient Influence Function for $\psi^{-1}$ derived in Lemma \ref{lemma:if_psi}. Intuitively, the one-step estimator is a debiased version of the plug-in estimator obtained from adding an estimate of the mean of the EIF to the plug-in estimator  (cf.  \cite{semiparametric, ehk_semiparametric}). Concretely, the proposed \textbf{\textit{one-step estimator}} $\widehat{\psi}^{-1}_{os}$ takes the form:
\begin{equation}
\label{eq:os_point_psi}
    \widehat{\frac{1}{\psi_{os}}} = \QQ_N \left( \left(\frac{1}{\widehat{\gamma}(X_i)} - 1\right)\sum_{(y_1, \dots, y_J) \neq 0} (-1)^{|y|+1} \frac{\mathbbm{1}(Y_i=y)}{\widehat{q}_y(X_i)} + 1\right).
\end{equation}
As before, to construct $\widehat{\psi}_{os}^{-1}$ we first estimate the nuisance $q$-probabilities from the available data, and then combine those into capture probabilities. Using these estimates, we then compute the empirical mean of the (uncentered) EIF to obtain $\widehat{\psi}_{os}^{-1}$.

\begin{remark}
    We assume the nuisance functions $q_y(X)$ are estimated using separate samples from the one used to compute the empirical mean $\QQ_N$ in the proposed estimators for $\psi^{-1}$. These cross-fitting/sample-splitting approaches control empirical process terms in our estimators, without introducing potentially restrictive assumptions \citep{chernozhukov2018double, ehk_semiparametric}. 
\end{remark}    

We now show that the proposed one-step estimator is nearly optimal, approximately normal, and exhibits a doubly-robust-type second order remainder. Specifically, we show that these properties hold for any number $N$ of observed units even when the $q$-probabilities are estimated flexibly using machine learning models. We restrict ourselves to providing finite sample guarantees because this is often the relevant setting in capture-recapture problems where it is not possible to let the sample size $N\to \infty$, as we aim to estimate a finite population of size $n \geq N$.

\begin{theorem} 
\label{thm:optimal}
For any sample size $N$ the proposed one-step estimator $\widehat{\psi}^{-1}_{os}$ satisfies
\begin{equation}
    \Big|(\widehat{\psi}^{-1}_{os} - \psi^{-1}) - \QQ_N(\phi)\Big| \leq \eta,
\end{equation}
for any $\eta > 0$, with probability at least
$$1 - \left(\frac{1}{\eta^2}\right) \E_{\QQ}\left[\widehat{R}_2^2 + \frac{1}{N}\norm{\widehat{\phi} - \phi}^2\right],$$
where $\phi$ is the Efficient Influence Function \eqref{eq:if_psi} with second-order reminder given by
\begin{equation}
\label{eq:r2_bound}
\begin{split}
    \widehat{R}_2 = \E_{\QQ}&\left[\left(\frac{1}{\widetilde{\gamma}(X)}-1\right)\left\{\sum_{y_i\neq y_j \neq 0} (-1)^{|y_i| + |y_j|}\frac{(q_{y_i}(X) - \widehat{q}_{y_i}(X))}{\widetilde{q}_{y_i}(X)}\frac{(q_{y_j}(X) - \widehat{q}_{y_j}(X))}{\widetilde{q}_{y_j}(X)}\right.\right. \\
    & + \left.\left. \sum_{|y|\neq 0, ~even} \left(\frac{q_{y}(X) - \widehat{q}_{y}(X)}{\widetilde{q}_{y}(X)}\right)^2\right\}\right],
\end{split}
\end{equation}
where the sums are taken with respect to capture profiles of the form $y = (y_1,\dots, y_J, 0, \dots, 0)$ and for some value $\widetilde{q}(X)$ of the $q$-probabilities between $q(X)$ and $\widehat{q}(X)$.
\end{theorem}

The preceding theorem shows that, for any given sample size $N$, the proposed one-step estimator $\widehat{\psi}^{-1}_{os}$ is nearly optimal, meaning that it approximates the empirical average of the EIF, $\QQ_N(\phi)$, centered at the true inverse capture probability. This approximation holds for any predetermined error tolerance $\eta > 0$, with high probability depending on the second-order remainder $\widehat{R}_2$ and the $L_2$-error in estimating the influence function itself. Crucially, both components are driven by the error in estimating the nuisance $q$-probabilities, but the latter to a lesser extent by the normalization by $N$. Note that Theorem \ref{thm:optimal} also provides a conventional convergence in probability result, since for any fixed $\eta$ the probability goes to $1$ whenever $\widehat{R}_2$ goes to $0$ as the sample size $N$ increases. 

Moreover, the previous result implies that the proposed estimator enjoys a finite-sample ``doubly-robust'' property in the sense that $\widehat{R}_2^2$ can be small even if some of the nuisance $q$-probabilities are not well estimated. In particular, our one-step estimator is nearly optimal with probability depending on the \textit{square} of the errors in estimating the $q_y$'s. For example, $R_{2}$ will be small (i.e., of order $1/\sqrt{N}$) if each of those estimation errors is of order $N^{-1/4}$, which can be achieved for many flexible nonparametric estimators \citep{ehk_semiparametric}. Thus, when a correct parametric model for the $q$-probabilities is not available, the proposed one-step estimator allows for the estimation of the nuisance functions using a data-driven approach while retaining desirable statistical guarantees. This result is detailed in the next Corollary.

\begin{corollary} 
\label{cor:doubly_robust}
Assume that $\widehat{q}_y(X) \geq \epsilon$ and $\norm{q_y(X) - \widehat{q}_y(X)} \leq \xi_N$, for all capture profiles $y = (y_1, \dots, y_{J}, 0, \dots, 0) \neq 0$. Then, 
\begin{equation}
    \PP\left(\Big|(\widehat{\psi}^{-1}_{os} - \psi^{-1}) - \QQ_N(\phi)\Big| \leq \eta\right) \geq 1 - \left(\frac{C}{\eta^2}\right) \left(\xi_N^2 + \frac{1}{N}\right)
\end{equation}
for any sample size N and any $\eta > 0$, where $C$ is a constant independent of N.
   
\end{corollary}

In summary, we have that the proposed one-step estimator  $\widehat{\psi}^{-1}_{os}$ is close to a sample average of the EIF with high probability when $\widehat{R}_2^2$ is small (nearly optimal), which can be achieved even in large nonparametric models. Notably, this approximation to a sample average of EIF values further implies that our estimator is approximately Gaussian, and therefore we can derive approximately valid confidence intervals using standard tools. This key result is formalized in the following Theorem.

\begin{theorem} 
\label{thm:normal}
$\widehat{\psi}^{-1}_{os} - \psi^{-1}$ follows an approximately Gaussian distribution, with the difference in cumulative distributions 
$$\Bigg| \mathbb{P}\left(\frac{\widehat{\psi}^{-1}_{os} - \psi^{-1}}{\widehat{\sigma}/\sqrt{N}} \leq t \right) - \Phi(t)\Bigg|$$
uniformly bounded above by
\begin{equation}
    \frac{C}{\sqrt{N}}\E_{\mathbb{Q}}\left[\frac{\rho}{\widetilde{\sigma}^3}\right] + \frac{1}{\sqrt{2\pi}}\left(\sqrt{N}\E_{\mathbb{Q}}\left[\frac{|\widehat{R}_2|}{\widetilde{\sigma}}\right] + |t|\E_{\mathbb{Q}}\left[\Big|\frac{\widehat{\sigma}}{\widetilde{\sigma}}-1\Big|\right] \right),
\end{equation}
where $\widehat{\sigma}^2 = \widehat{var}(\widehat{\phi})$,     $\widetilde{\sigma}^2 = var(\widehat{\phi}|Z^n)$, $\rho = \E[|\widehat{\phi}  - \QQ(\widehat{\phi})|^3 | Z^n]$, and $C<1/2$ is the Berry-Esseen constant.
\end{theorem}

The previous Berry-Esseen-type result implies that usual ($1-\alpha$)-level confidence intervals of the form 
\begin{equation}
\label{eq:ci_psi}
    \widehat{CI}(\psi^{-1}) = \left[\widehat{\psi}^{-1}_{os} \pm z_{1-\alpha/2} \widehat{\sigma}/\sqrt{N}\right],
\end{equation}
can be used for inference on the inverse capture probability with nearly-valid finite sample coverage guarantees, where $\widehat{\sigma}^2 = \widehat{var}(\widehat{\phi})$ is the unbiased empirical variance of the estimated efficient influence function. Therefore, using the proposed identification strategy, the derived one-step estimator for the capture probability is not only nearly optimally efficient, but is also approximately Normal in finite samples. The first property means that it approximates a sample average of EIF values with high probability if the nuisance errors are small in a doubly-robust fashion, while the latter provides a general approach for inference on $\psi^{-1}$ and to construct confidence intervals with non-asymptotic guarantees in a Berry-Esseen sense. 

Finally, given that our identification assumption generalizes the two predominant approaches in the capture-recapture literature, the derived one-step estimator can be adapted to these settings while retaining its robust statistical properties. For instance, when $J=K$, the proposed one-step estimator can be used to estimate the inverse capture probability by considering all capture profiles $y=(y_1,\dots, y_K)$ and corresponding $q$-probabilities $\QQ(Y=y \mid X)$. In the $J=2$ setting, our doubly-robust estimator simplifies to 
\begin{equation}
\label{eq:l2_condind}
        \widehat{\frac{1}{\psi_{os}}} = \QQ_N \left( \left(\frac{1}{\widehat{\gamma}(X)} - 1\right)\left\{ \frac{Y_1(1-Y_2)}{\widehat{q}_{10}(X)} + \frac{(1-Y_1)Y_2}{\widehat{q}_{01}(X)} - \frac{Y_1Y_2}{\widehat{q}_{11}(X)} \right\} + 1\right),
\end{equation}
where $\widehat{q}_{y_1y_2}(X)$ are estimates of $q_{y_1y_2}(X) = \QQ(Y_1=y_1, Y_2=y_2,Y_3=0, \dots, Y_K=0 \mid X)$.
\section{Sensitivity Analysis: Bounded Highest-Order Interaction}

In this section, we relax the no highest-order interaction assumption and propose a sensitivity analysis that allows for arbitrary dependence between the $J$ capture episodes in our conditional log-linear model among the unobserved \eqref{eq:log_lineal}. Specifically, relaxing the identifying assumption leads to a \textit{partial identification} of $\psi^{-1}$, resulting in an interval that contains the true target parameter as a function of a parameter bounding the highest-order interaction coefficient. We then leverage the one-step estimator derived for the no highest-order interaction case to propose similar one-step estimators for the upper and lower bounds of interest. Our estimators are shown to be nearly optimal and approximately normal, even when flexibly estimating the $q$-probabilities, under an additional \textit{margin condition} \eqref{eq:margin_cond}. Finally, we derive approximately valid confidence intervals that can be combined to bound the true capture probability target parameter with high confidence. The proposed sensitivity analysis reveals how deviations from the no highest-order interaction and positivity assumptions impact the resulting population size estimates, allowing analysts to incorporate domain-specific knowledge into the identification and estimation steps more transparently and rigorously.

\subsection{Partial identification}

We relax the identifying assumption of no highest-order interaction in our saturated conditional log linear model among the unobserved \eqref{eq:log_lineal} and replace it by the assumption that the highest-order interaction coefficient, $\alpha_1(X)$, is uniformly bounded by a parameter $\delta \geq 0$. Mathematically, in a log-linear model for $J$ lists conditional on covariates and on being unobserved by the remaining $K-J$ lists, the relaxed identification approach assumes that
\begin{equation}
     \Big|\alpha_1(X)\Big| = \Bigg| \sum_{(y_1, \dots, y_J) \neq 0} (-1)^{J + |y|}\log(q_y(X)\gamma(X)) + (-1)^{J}\log(1-\gamma(X))  \Bigg| \leq \delta,
\end{equation}
where $|y| = \sum_{j=1}^{J} y_j$ is the number of positive lists in the capture profile $y = (y_1, \dots, y_J, 0, \dots, 0)$. Note that $\delta = 0$ recovers the no highest-order interaction assumption and the previously proposed estimator would still be valid. On the other hand, setting $\delta > 0$ is indeed a relaxation of our assumption allowing for a non-negligible highest-order interaction between the $J$ capture episodes. The Appendix provides an interpretation of the bounded highest-order interaction assumption in terms of conditional odds ratios.

By setting a bound on the highest-order interaction parameter, this approach yields \textit{bounds} on the inverse capture probability $\psi^{-1}$. Essentially, this means that while we can no longer identify the exact value of the capture probability, we achieve a \textit{partial identification} of $\psi$ using an interval depending on the bound $\delta$. 
From the relaxed identifying assumption $|\alpha_1(X)|\leq \delta$, we obtain a partial identification expression detailed in the following Proposition \ref{prop:partial_psi}.

\begin{proposition} 
\label{prop:partial_psi}
\textbf{(Partial-Identification Result)} Assume positivity and bounded highest-order interaction coefficient, $|\alpha_{1}(X)| \leq \delta$, in a log-linear model for $J$ lists, conditional on covariates and on not being captured by the other $K-J$ lists. Then, we can partially identify the inverse capture probability $\psi^{-1}$ using the observed data with the lower and upper bounds
\begin{equation}
\label{eq:lower_psi}
    \frac{1}{\psi_{\ell}} = \E_{\QQ}\left[\min \left\{ 1 + \exp\left(\sum_{(y_1, \dots, y_J) \neq 0} (-1)^{|y| + 1}\log(q_y(X)) - \delta \right), \frac{1}{\epsilon} \right\} \right],
\end{equation}
\begin{equation}
\label{eq:upper_psi}
    \frac{1}{\psi_{u}} = \E_{\QQ}\left[\min \left\{ 1 + \exp\left(\sum_{(y_1, \dots, y_J) \neq 0} (-1)^{|y| + 1}\log(q_y(X)) + \delta \right), \frac{1}{\epsilon} \right\} \right],
\end{equation}
such that $\psi^{-1} \in \left[\psi_{\ell}^{-1}, \psi_{u}^{-1}\right]$. As before, $y$ denotes the complete $K$-dimensional capture profile $y = (y_1, \dots, y_J, 0, \dots, 0)$.
\end{proposition}

Intuitively, the proposed bounds are derived from a partial identification expression for the conditional capture probability
\begin{equation*}
\label{eq:part_gamma}
    \frac{1}{\gamma(X)} \in \left[1 + \exp\left(\sum_{(y_1, \dots, y_J) \neq 0} (-1)^{|y| + 1}\log(q_y(X)) \pm \delta \right) \right],
\end{equation*}
with an additional constraint on the parameter $\delta$ to ensure that $\gamma(X) \in [\epsilon, 1]$, under the positivity assumption on the conditional capture probability. This parametrization of the identification assumptions allows to examine how variations in the $\delta$ and $\epsilon$ parameters can affect estimates of the total population size. This provides a straightforward and flexible approach to integrating domain-specific insights into the identification and estimation procedures.

\subsection{One-step estimator}

We now turn to the problem of estimating the proposed bounds that partially identify the capture probability under the bounded highest-order interaction condition. With a slight abuse of notation, we use $\psi_{\delta}^{-1}$ to refer to either the lower or upper bound derived in Proposition \ref{prop:partial_psi}. By their symmetry, the results presented in this section apply to both $\psi_{\ell}^{-1}$ and $\psi_{u}^{-1}$, so we present them succinctly as applying to the generic bound $\psi_{\delta}^{-1}$ defined by
$$\frac{1}{\psi_{\delta}} = \E_{\QQ}\left[\min \left\{\frac{1}{\gamma_{\delta}(X)}, \frac{1}{\epsilon} \right\} \right] = \E_{\QQ}\left[\min \left\{ 1 + \exp\left(\sum_{(y_1, \dots, y_J) \neq 0} (-1)^{|y| + 1}\log(q_y(X)) \pm \delta \right), \frac{1}{\epsilon} \right\} \right],$$
where $\gamma_{\delta}(X)^{-1} = 1 + \exp\left(\sum_{(y_1, \dots, y_J) \neq 0} (-1)^{|y| + 1}\log(q_y(X)) \pm \delta \right)$ denotes the conditional capture probability associated with the corresponding bound of interest and $y=(y_1, \dots, y_J, 0, \dots, 0)$.

Our previous strategy was to use the Efficient Influence Function for $\psi^{-1}$ to construct a one-step estimator for the capture probability under the assumption of no highest-order interaction. However, since the min function appearing in the identifying expressions in Proposition \ref{prop:partial_psi} is not smooth, we cannot directly follow the same approach and derive the EIF to propose estimators for the target bounds $\psi^{-1}_{\ell}$ and $\psi^{-1}_u$. To overcome this, note that we can rewrite our generic bound as
\begin{equation}
\label{eq:bound_indicator}
\begin{split}
    \frac{1}{\psi_{\delta}} & = \E_{\QQ}\left[\frac{1}{\gamma_\delta(X)}\mathbbm{1}\left(\frac{1}{\gamma_\delta(X)} \leq \frac{1}{\epsilon} \right)\right] + \frac{1}{\epsilon}\E_{\QQ}\left[\mathbbm{1}\left(\frac{1}{\gamma_\delta(X)} > \frac{1}{\epsilon} \right)\right] \\ 
    & = \E_{\QQ}\left[\left(\frac{1}{\gamma_\delta(X)} - \frac{1}{\epsilon}\right)\mathbbm{1}\left(\frac{1}{\gamma_\delta(X)} - \frac{1}{\epsilon} \leq 0 \right)\right] + \frac{1}{\epsilon},
\end{split}
\end{equation}
which suggest to use a flexible estimator for the first term involving the conditional capture probability, and a plug-in estimator for the indicator function to deal with its non-smoothness. 

Concretely, we propose the following \textbf{one-step estimator} for $\psi_{\delta}^{-1}$:
\begin{equation}
\label{eq:onestep_bound}
    \widehat{\frac{1}{\psi_{\delta}}} = \QQ_N\left(\left(\widehat{\varphi}_{\delta}(X, Y)-\frac{1}{\epsilon}\right)\mathbbm{1}\left(\frac{1}{\widehat{\gamma}_\delta(X)} - \frac{1}{\epsilon} \leq 0 \right) \right) + \frac{1}{\epsilon},
\end{equation}
where 
\begin{equation}
\label{eq:if_delta}
    \varphi_{\delta}(X,Y) = \left(\frac{1}{\gamma_{\delta}(X)} - 1\right)\sum_{(y_1, \dots, y_J) \neq 0} (-1)^{|y|+1} \frac{\mathbbm{1}(Y=y)}{q_y(X)} + 1,
\end{equation}
is the (uncentered) EIF for $\E[\gamma_{\delta}(X)^{-1}]$. 
Note that the EIF \eqref{eq:if_delta} parallels that in Lemma \ref{lemma:if_psi} for $\psi^{-1}$ in the no highest-order interaction assumption setting. This is a consequence of the analogous conditional capture probabilities $\gamma(X)$ and $\gamma_{\delta}(X)$ involved in the derivation of the corresponding EIFs. This similarity also endows the proposed one-step estimator $\widehat{\psi}_{\delta}^{-1}$ with statistical properties akin to those of the one-step estimator derived for the scenario with no highest-order interaction.

Specifically, in the following we show that the proposed estimator is nearly optimal and approximately Gaussian with finite sample guarantees. Intuitively, the proposed one-step estimator behaves like an \textit{infeasible} estimator that has full knowledge of the conditional capture probability $\gamma_{\delta}(X)$ and uses it to compute the indicator function in \eqref{eq:bound_indicator}. For such ``oracle'' behavior to be met, a sufficient condition is the so-called \textit{margin condition}, which has been used to deal with non-smooth estimands in causal inference  \citep{luedtke2016statistical, kennedy2020sharp, levis2023covariate} and classification settings \citep{audibert2007tsybakov}.  Formally, we assume that there exists $\beta > 0$ and a nonnegative constant $C$, such that for any $t \geq 0$
\begin{equation}
\label{eq:margin_cond}
    \PP\left(\Bigg|\frac{1}{\gamma_\delta(X)} - \frac{1}{\epsilon} \Bigg| \leq t \right) \leq Ct^{\beta}.
\end{equation}
The margin condition controls the probability that the conditional capture probability is near the bound $\frac{1}{\epsilon}$ making the indicator function non-differentiable, with $\beta$ controlling how quickly this probability vanishes and becomes negligible for the correct estimation of $\psi_\delta^{-1}$. This key result is formalized in the following Theorem.

\begin{theorem} 
\label{thm:optimal_bounds}
Under the margin condition for $\beta > 0$ and nonnegative $C$, the one-step estimator $\widehat{\psi}^{-1}_{\delta}$ satisfies 
\begin{equation}
    \PP\left(\Big|(\widehat{\psi}^{-1}_{\delta} - \psi^{-1}_{\delta}) - \QQ_N(\phi_\delta)\Big| \leq \eta\right)  \geq 1 - \frac{1}{\eta^2}\E\left[\widehat{R}_{2,\delta}^2 + \frac{1}{N}\norm{\widehat{\phi}_{\delta} - \phi_{\delta}}^2\right]
\end{equation}
for any sample size $N$ and $\eta > 0$, where \begin{equation}
\phi_{\delta}(X,Y) = \left(\varphi_{\delta}(X,Y)-\frac{1}{\epsilon}\right)\mathbbm{1}\left(\frac{1}{\gamma_\delta(X)} - \frac{1}{\epsilon} \leq 0 \right)  + \frac{1}{\epsilon} - \frac{1}{\psi_{\delta}}
\label{eq:if_bound}
\end{equation}
and the reminder term is bounded by 
\begin{equation}
\label{eq:r2_delta}
\widehat{R}_{2,\delta} \leq \sum_{y_i\neq y_j \neq 0} \norm{q_{y_i} - \widehat{q}_{y_i}} \norm{q_{y_j} - \widehat{q}_{y_j}} + \sum_{|y|\neq 0, ~even} \norm{q_{y} - \widehat{q}_{y}}^2 + C\norm{\widehat{\gamma} - \gamma}_{\infty}^{1+\beta},
\end{equation}
where the sums are taken with respect to capture profiles $y = (y_1,\dots, y_J,0, \dots, 0)$.
\end{theorem}

Similarly to the setting assuming no highest-order interaction, this result shows that the proposed one-step estimator for the generic bound $\widehat{\psi}_\delta^{-1}$ is close to the empirical average of its corresponding EIF \eqref{eq:if_bound}. This holds for any given error tolerance $\eta > 0$ and sample size $N$ with high probability. Crucially, this probability now depends on the margin condition, in addition to the second-order remainder $\widehat{R}_2$ and the $L_2$-error in estimating the EIF. Therefore, the proposed estimator enjoys a similar doubly-robust-type property if $\beta \geq 1$, depending on the \textit{square} of the errors in estimating the $q_y$'s. Analogous to Corollary \ref{cor:doubly_robust}, a sufficient condition for $R_{2} = o_{\PP}(1/\sqrt{N})$ is that each of the estimation errors $\norm{\widehat{q}_y(X) - q_y(X)}$ be of order $N^{-1/4}$, which can be obtained for many off-the-shelf machine learning models \citep{ehk_semiparametric}. This allows to estimate the nuisance functions in a data-driven way without sacrificing desired statistical properties. Note that the previous theorem also gives standard conditions for $\widehat{\psi}_{\delta}^{-1}$ to be $\sqrt{N}$-consistent and asymptotically normal.


Moreover, Theorem \ref{thm:optimal_bounds} provides that the proposed one-step estimator $\widehat{\psi}^{-1}_{\delta}$ is not only close to a sample mean of EIF values, but is also approximately Gaussian in finite samples. We detail this result in the following Theorem and use such property to derive approximately valid confidence intervals for the inverse capture probability.

\begin{theorem} 
\label{thm:normal_bounds}
$\widehat{\psi}^{-1}_{\delta} - \psi^{-1}_{\delta}$ follows an approximately Gaussian distribution, with the difference in cumulative distributions uniformly bounded above by
$$\Bigg| \mathbb{P}\left(\frac{\widehat{\psi}^{-1}_{\delta} - \psi^{-1}_{\delta}}{\widehat{\sigma}_{\delta}/\sqrt{N}} \leq t \right) - \Phi(t)\Bigg| \leq
\frac{1}{\sqrt{2\pi}}\left(\sqrt{N}\E_{\mathbb{Q}}\left[\frac{|\widehat{R}_{2,\delta}|}{\widetilde{\sigma}_{\delta}}\right] + |t|\E_{\mathbb{Q}}\left[\Big|\frac{\widehat{\sigma}_{\delta}}{\widetilde{\sigma}_{\delta}}-1\Big|\right] \right) + \frac{C}{\sqrt{N}}\E_{\mathbb{Q}}\left[\frac{\rho_{\delta}}{\widetilde{\sigma}^3_{\delta}}\right],$$
where $\widehat{\sigma}^2_{\delta} = \widehat{var}(\widehat{\phi}_{\delta})$,     $\widetilde{\sigma}^2_{\delta} = var(\widehat{\phi}_{\delta}|Z^n)$, $\rho_{\delta} = \E[|\widehat{\phi}_{\delta}  - \QQ(\widehat{\phi}_{\delta})|^3 | Z^n]$, $C<1/2$ is the Berry-Esseen constant, and $\widehat{R}_{2,\delta}$ is bounded by \eqref{eq:r2_delta}.
\end{theorem}

The previous theorem implies that the following are approximately valid $(1-\alpha)$-level confidence intervals for the bounds on the capture probability $\psi^{-1} \in [\psi^{-1}_{\ell}, \psi^{-1}_u]$:
\begin{equation}
\label{eq:cin_psi_bounds}
    \widehat{CI}(\psi^{-1}_{\ell}) = \left[\widehat{\psi}^{-1}_{\ell} \pm z_{1-\alpha/2} \widehat{\sigma}_{\ell}/\sqrt{N}\right], \quad\quad
    \widehat{CI}(\psi^{-1}_{u}) = \left[\widehat{\psi}^{-1}_{u} \pm z_{1-\alpha/2} \widehat{\sigma}_{u}/\sqrt{N}\right],
\end{equation}
where $\widehat{\psi}_{\ell}^{-1}$ and $\widehat{\psi}_{u}^{-1}$ are the proposed one-step estimators \eqref{eq:onestep_bound}, and $\widehat{\sigma}^2_{\ell} = \widehat{var}(\widehat{\phi_{\ell}})$, $\widehat{\sigma}^2_{u} = \widehat{var}(\widehat{\phi_{u}})$, are the unbiased empirical variances of the corresponding EIFs estimated from expression \eqref{eq:if_bound}.  Therefore, an asymptotically valid ($1-\alpha$)-level confidence interval for $\psi^{-1}$ can be derived by combining the lower limit of the CI for $\widehat{\psi}^{-1}_{\ell}$ and the upper limit of the CI for $\widehat{\psi}^{-1}_{u}$:
\begin{equation}
    \psi^{-1} \in \left[\widehat{\psi}_{\ell}^{-1} - z_{1-\alpha/2}\widehat{\sigma}_{\ell}/\sqrt{N}, \quad \widehat{\psi}_{u}^{-1} + z_{1-\alpha/2}\widehat{\sigma}_{u}/\sqrt{N} \right],
\end{equation}
where $\{\widehat{\sigma}_{\ell}, \widehat{\sigma}_{u}\}$ are defined as before. Alternatively, for more precise inference, the CIs in \eqref{eq:cin_psi_bounds} can be combined into a confidence interval for the partially identified parameter $\psi^{-1}$ using the methods proposed in \cite{imbens2004confidence}.

In summary, the proposed sensitivity analysis accommodates a more general dependence between capture episodes in a log-linear model at the expense of only achieving partial identification of the capture probability. We argue that the proposed bounds are still relevant to population size estimation because they provide an informative interval containing the true capture probability, relying on a weaker identification assumption than the no highest-order interaction. Moreover, the $\delta$ and $\epsilon$ parameters allow analysts to incorporate domain knowledge about possible bounds on the highest-order interaction and positivity assumptions, and facilitates the examination of how variations in these bounds can affect estimates of the total population size. As before, our approach can also be used for sensitivity analysis in standard log linear models ($J=K$) and conditional independence settings ($J=2$). 


\section{Inference for total population size}

So far, we have focused on constructing flexible and robust estimators for the the inverse capture probability $\psi^{-1}$ with finite sample guarantees. In this section, we shift our attention to estimating the total \textit{population size} $n$, which is typically the ultimate goal in a capture-recapture setting. To achieve this, we use $\widehat{n} = N/\widehat{\psi}$ as natural estimator for $n$, given an estimator for the capture probability $\widehat{\psi}$, where $N \sim \text{Bin}(n, \psi)$ is the observed sample size of uniquely captured subjects. Therefore, in the following we leverage the one-step estimators proposed above and provide estimators and approximately valid confidence intervals for $n$.  

\subsection{Population size under no highest-order interaction}

Concretely, in the no highest-order interaction scenario, we propose to estimate the total population size $n$ via $\widehat{n}_{os} = N/\widehat{\psi}_{os}$, where $\widehat{\psi}_{os}^{-1}$ is the one-step estimator proposed in \eqref{eq:os_point_psi}. We then derive a confidence interval for the total population size $n$ using Theorem 4 in \cite{das2023doubly}, provided that $\widehat{\psi}_{os}^{-1}$ is approximated by a sample average. We adapt this result to our setting in the following Theorem \ref{thm:ci_n}.

\begin{theorem}
    Consider a log-linear model for $J$ lists, conditional on covariates and on not being captured by the remaining $K-J$ lists. Assuming no highest-order interaction term, the following is an approximately $(1-\alpha)$ confidence interval for the population size $n$ centered at $\widehat{n}_{os} = N/\widehat{\psi}_{os}$ 
\begin{equation}
\label{eq:ci_n}
    \widehat{CI}(n) = \left[ \widehat{n}_{os} \pm z_{\alpha/2}\sqrt{\widehat{n}_{os}\left(\widehat{\psi}_{os}\widehat{\sigma}^2 + \frac{1-\widehat{\psi}_{os}}{\widehat{\psi}_{os}} \right)}\right],
\end{equation}
where $\widehat{\sigma}^2 = \widehat{var}(\widehat{\phi})$ is the unbiased empirical variance of the estimated EIF \eqref{eq:if_psi} and $\widehat{\psi}_{os}^{-1}$ is the proposed one-step estimator \eqref{eq:os_point_psi}.    
\label{thm:ci_n}
\end{theorem}

\subsection{Bounds on the population size under partial identification}

The relaxed assumption of a bounded highest-order interaction coefficient in the saturated conditional log-linear model among the unobserved \eqref{eq:log_lineal} implies lower and upper \textit{bounds} on the total population size $n$ given by
\begin{equation}
    n_{\ell} = \frac{N}{\psi_{\ell}}, \quad\quad\quad n_{u} = \frac{N}{\psi_{u}},
\end{equation}
such that $n \in [n_{\ell}, n_u]$. Analogously as before, we estimate such bounds using the derived one-step estimators \eqref{eq:onestep_bound} for $\psi_{\ell}^{-1}$ and $\psi_{u}^{-1}$, respectively. Furthermore, in the following Theorem, we leverage that our one-step estimators can be approximated by sample averages and derive approximately valid confidence intervals for the estimated bounds on the population size $n$.

\begin{theorem}
\label{thm:cin_bounds}
Consider a log-linear model for $J$ lists, conditional on covariates and on not being captured by the remaining $K-J$ lists. The following are approximate $(1-\alpha)$ confidence intervals for the bounds on the population size $n \in [n_{\ell}, n_u]$, assuming a bounded highest-order interaction term: 
\begin{equation}
\label{eq:cin_upper}
    \widehat{CI}(n_{\ell}) = \left[ \frac{N}{\widehat{\psi}_{\ell}} \pm z_{\alpha/2}\sqrt{N \left(\widehat{\sigma}^2_{\ell} + \frac{1-\widehat{\psi}_{\ell}}{\widehat{\psi}_{\ell}^2} \right)}\right],
\end{equation}
\begin{equation}
\label{eq:cin_lower}
    \widehat{CI}(n_{u}) = \left[ \frac{N}{\widehat{\psi}_{u}} \pm z_{\alpha/2}\sqrt{N \left(\widehat{\sigma}^2_{u} + \frac{1-\widehat{\psi}_{u}}{\widehat{\psi}_{u}^2} \right)}\right],
\end{equation}
where $\widehat{\psi}_{\ell}$ and $\widehat{\psi}_{u}$ are one-step estimators \eqref{eq:onestep_bound} for the corresponding lower and upper bounds in Proposition \ref{prop:partial_psi}, and $\widehat{\sigma}^2_{\ell} = \widehat{var}(\widehat{\phi_{\ell}})$, $\widehat{\sigma}^2_{u} = \widehat{var}(\widehat{\phi_{u}})$, are the unbiased empirical variances of the estimated influence functions derived in \eqref{eq:if_bound}.     
\end{theorem}

Finally, the previous theorem implies that 
\begin{equation}
\label{eq:cin_bounds}
    \widehat{CI}(n) = \left[ \frac{N}{\widehat{\psi}_{\ell}} - z_{\alpha/2}\sqrt{N \left(\widehat{\sigma}^2_{\ell} + \frac{1-\widehat{\psi}_{\ell}}{\widehat{\psi}_{\ell}^2} \right)}, \quad 
    \frac{N}{\widehat{\psi}_{u}} + z_{\alpha/2}\sqrt{N \left(\widehat{\sigma}^2_{u} + \frac{1-\widehat{\psi}_{u}}{\widehat{\psi}_{u}^2} \right)} \right]
\end{equation}
is an approximate $(1-\alpha)$ confidence interval for the total population size $n$. Similarly as before, these bound estimates can be combined for more precise inference following \cite{imbens2004confidence} methodology.

Note that $\delta = 0$ recovers the CI \eqref{eq:ci_n} for the corresponding point-identified population size $n$. In both cases, the coverage guarantee of the derived confidence intervals is mainly driven by the error term through $\E[|\widehat{R}_2^2|]$, under standard boundedness assumptions. This highlights the importance of efficient estimation of $\{\psi_{\ell}, \psi_{u}\}$ (or $\psi$ under no highest-order interaction) and the approximation of $\{\widehat{\psi}_{\ell}, \widehat{\psi}_{u}\}$ ($\widehat{\psi}$) by sample averages. Crucially, our proposed one-step estimators can attain an error $\widehat{R}_2^2$ of order $1/\sqrt{N}$ as shown above, even when the $q$-probabilities are estimated flexibly using machine learning models.

\section{Data Analysis: Synthetic and Real Data from the Peruvian Conflict}

In the following, we demonstrate the advantages of the proposed methods using both simulated and real capture-recapture data from the Peruvian internal armed conflict (1980-2000). We compare our one-step estimator with a simple, but typically suboptimal, plug-in estimator and show that our method dominates in several regimes. Moving forward, we use data from the Peruvian Truth Commission to estimate the number of casualties and highlight the role of the identification strategy adopted on the resulting estimated population size. Finally, we apply the proposed sensitivity analysis to estimate bounds on the population size for different values of the highest-order interaction term bound $\delta$, providing intervals that partially identify the number of casualties and illustrating the sensitivity of the estimates to the identification assumptions.

\subsection{Synthetic capture-recapture data with 3 lists and no highest-order interaction assumption}

To empirically demonstrate the properties of the proposed one-step estimator \eqref{eq:os_point_psi}, we begin by simulating a scenario with latent heterogeneous capture probabilities $q_y$ drawn from a uniform distribution 
with an empirical lower bound of approximately $0.05$. 
We then simulate the estimation step of the $q$-probabilities by setting
$$\widehat{q}_y = expit(logit(q_y) + \epsilon_y), ~~~~~ \epsilon_y \sim \mathcal{N}(bn^{-\alpha}, n^{-2\alpha}),$$
where the $\alpha$ term gives the rate of convergence of the estimated $\widehat{q}_y$ to the true latent $q_y$ (because the RMSE scales as $n^{-\alpha}$), and $b$ controls the bias of the estimates. We vary $\alpha$  and $b$ to analyze the behavior of the resulting $\widehat{\psi}^{-1}$ for different biases and rates of convergence of the nuisance functions, simulating different scenarios where an analyst is able to estimate the $q$-probabilities with varying precision. All the simulations are done $1,000$ times for a true $n=10,000$ and $\psi^{-1} = 1.41$ ($\psi = 0.7$), yielding an observed $N \approx 7,000$.

From our experiments we can observe that, when the $q$-probabilities are estimated at better rates closer to $\alpha = 0.5$ (the parametric rate), both the plug-in and the one-step estimators perform well with decreasing bias and MSE, and coverage at the expected nominal level (Figures \ref{fig:cr_small_noise} and \ref{fig:cr_large_noise}). The advantages of the proposed one-step estimator are larger in regimes of slow rate of convergence, especially for $\widehat{q}_y$ estimates with large noise $b=10$ (Fig. \ref{fig:cr_large_noise}). For instance, for $\alpha$ smaller than the non-parametric rate of $\alpha=0.25$, the plug-in estimator $\widehat{\psi}_{pi}^{-1}$ has worst bias, MSE, and coverage than our one-step estimator $\widehat{\psi}_{os}^{-1}$. Note that the plug-in estimator does not have a well-defined variance formula, thus we use the estimated variance from the one-step estimator to construct corresponding confidence intervals. Finally, it should also be noticed that in the large noise regime the one-step estimator suffers from large variance which dominates the behavior of its MSE. 

\begin{figure}[ht]

    \centering
    \begin{subfigure}{0.3\textwidth}
        \includegraphics[width=\linewidth]{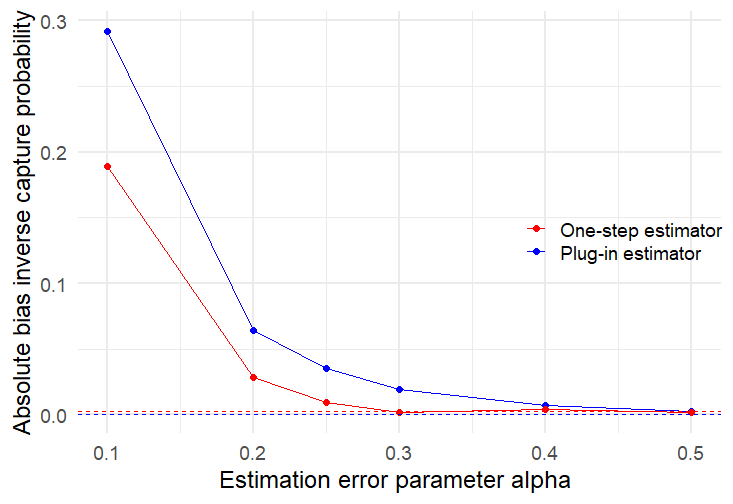}
        \caption{Bias $|\psi^{-1} - \widehat{\psi}^{-1}|$.}
    \end{subfigure}
    \hfill
    \begin{subfigure}{0.3\textwidth}
        \includegraphics[width=\linewidth]{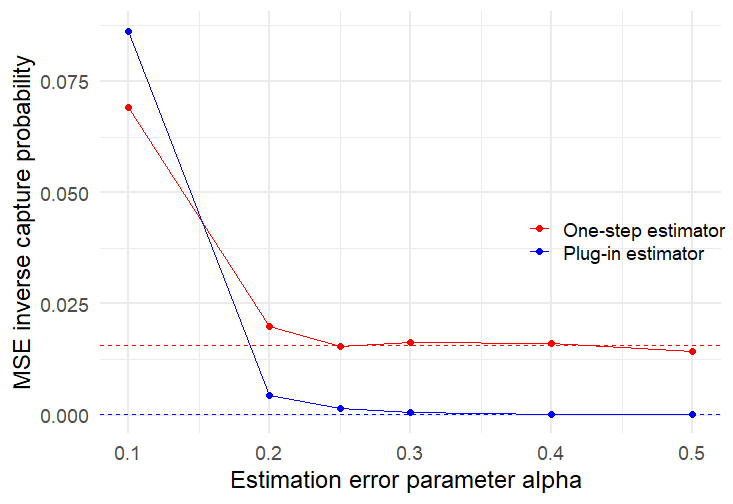}
        \caption{MSE $(\psi^{-1} - \widehat{\psi}^{-1})^2$.}
    \end{subfigure}
    \hfill
    \begin{subfigure}{0.3\textwidth}
        \includegraphics[width=\linewidth]{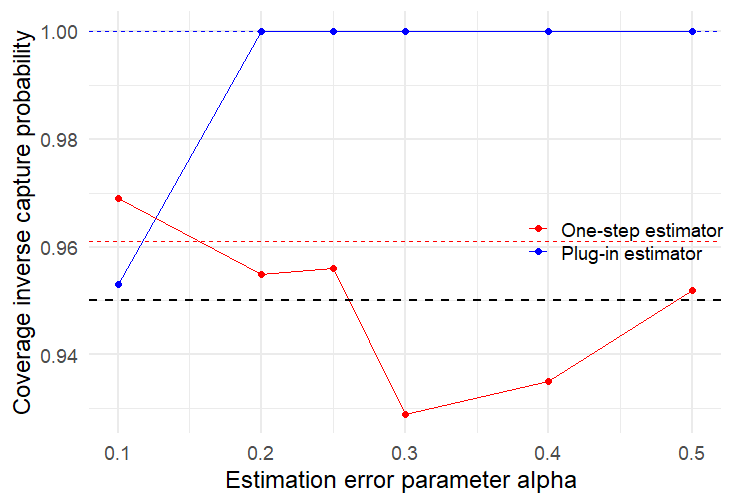}
        \caption{Empirical coverage of $\psi^{-1}$.}
    \end{subfigure}

    \caption{Results for $\psi^{-1}$ using synthetic data from three lists and small noise on $\widehat{q}_y$ ($b=1)$. True $n=10,000$ and $\psi^{-1} = 1.41$. Average over $1,000$ simulations.}
    \label{fig:cr_small_noise}
\end{figure}
\begin{figure}[ht]
    \centering
    \begin{subfigure}{0.3\textwidth}
        \includegraphics[width=\linewidth]{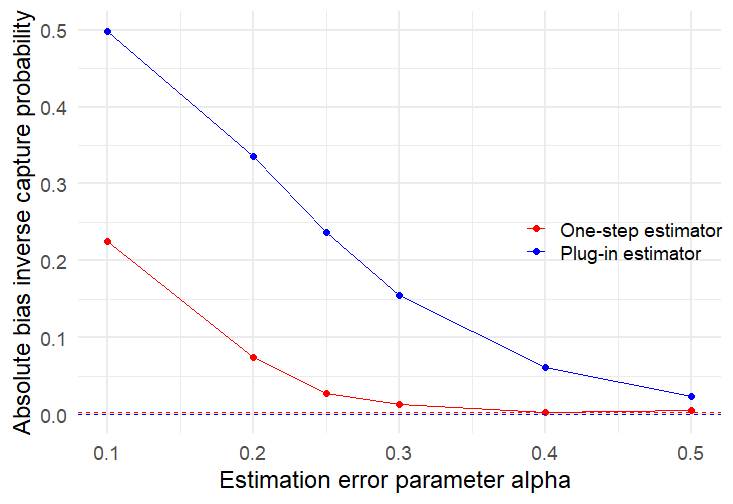}
        \caption{Bias $|\psi^{-1} - \widehat{\psi}^{-1}|$.}
    \end{subfigure}
    \hfill
    \begin{subfigure}{0.3\textwidth}
        \includegraphics[width=\linewidth]{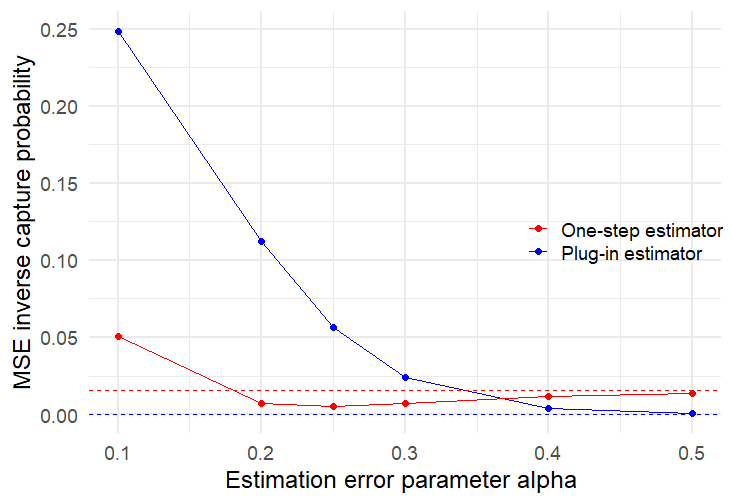}
        \caption{MSE $(\psi^{-1} - \widehat{\psi}^{-1})^2$.}
    \end{subfigure}
    \hfill
    \begin{subfigure}{0.3\textwidth}
        \includegraphics[width=\linewidth]{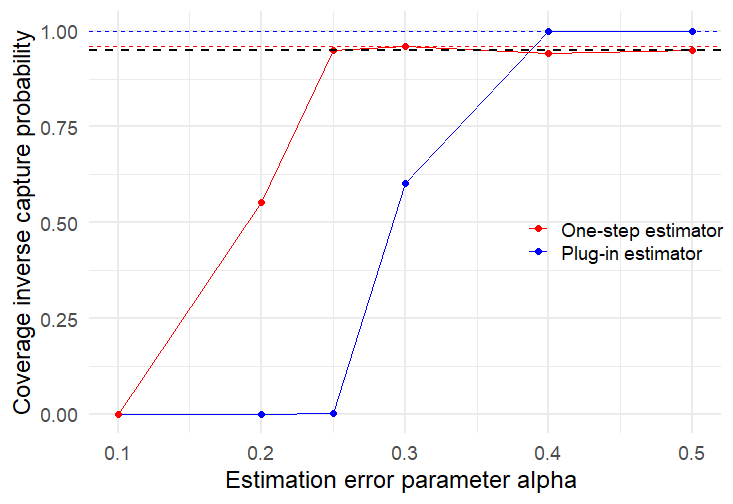}
        \caption{Empirical coverage of $\psi^{-1}$.}
    \end{subfigure}
    \caption{Results for $\psi^{-1}$ on synthetic data from three lists and large noise on $\widehat{q}_y$ ($b=10)$. True $n=10,000$ and $\psi^{-1} = 1.41$. Average over $1,000$ simulations.}
    \label{fig:cr_large_noise}
\end{figure}

Overall, the proposed one-step estimator for the (inverse) capture probability $\psi^{-1}$ exhibits a better performance than the plug-in estimator, especially when the nuisance $q$-probabilities cannot be well estimated. This is often the case in the absence of a reliable parametric model for $q_y$ and in complex capture-recapture settings where the capture probabilities vary with respect to continuous or high-dimensional variables. In these scenarios, the proposed one-step estimator proves useful allowing for flexible estimation of the $q$-probabilities and robust estimation of $\psi$.

\subsection{Application to the Peruvian internal armed conflict} 

Building on the theoretical and empirical advantages of our one-step estimator, we apply the proposed methodology to estimate the size of the victim population in the Peruvian internal armed conflict. This capture-recapture dataset was used by the Peruvian Truth and Reconciliation Commission to estimate the number of victims of killings and disappearances during the conflict between 1980 and 2000 \citep{patrick2003many}. The dataset consists of seven lists coming from three different sources of information documenting the conflict: (i) the Public Defender's Office (CP), (ii) the Truth and Reconciliation Commission (CVR), and (iii) five human rights and non-governmental organizations (ODH). The data includes information on $24,692$ uniquely identified victims of killings or disappearances, their capture profiles, and complementary demographic variables \citep{patrick2003many, rendon2019capturing}.

The original report by the Peruvian Truth and Reconciliation Commission estimated approximately 69,000 victims during the conflict using log-linear models and the no highest-order interaction assumption \citep{patrick2003many}. However, recent academic articles have estimated the victim population size in 48,000 using an alternative methodological approach \citep{rendon2019capturing}, between 58,000 and 66,000 using Bayesian latent class models \citep{manrique2019estimating}, and in 69,000 victims using a two-list doubly robust estimator under conditional independence \citep{das2023doubly}. Following \cite{das2023doubly}, we estimate the $q$-probabilities from the available demographic variables of sex, age and geographical location using the \textit{Random Forest} methodology \citep{breiman2001random}. 

First, we follow the approach in previous literature and consolidate the five lists from different NGOs into a single list (ODH in \citep{patrick2003many}) and use a standard conditional log-linear model on the resulting $K=3$ lists: CP, CVR, and ODH. We analyze how the estimated victim population size varies for different positivity lower bounds $\epsilon$ truncating the estimated $\widehat{q}_y$'s in Figure \ref{fig:peru_pos}. Without imposing a lower bound on the estimated $q$-probabilities (i.e., no positivity constraints), we obtain uninformative impossible estimates of $\widehat{n} = 30'715,584$ ($95\%$ CI: $[24,692, 79'973,134]$), larger than the Peruvian population in 1980. In this setting, the plug-in estimator returns $\widehat{n} = \infty$. Notably, by imposing a $1\%$ lower bound on all the $q$-probabilities, the proposed one-step estimator estimates $\widehat{n} = 121,733$ ($95\%$ CI: $[24,692, 456,336]$),\footnote{The provided CI is not symmetric around $\widehat{n}$ since we set the lower bound to be the observed $N$ whenever it is estimated to be a lower value.} while using $10\%$ as the lower bound yields $\widehat{n} = 52,215$ ($95\%$ CI: $[51,073, 53,356]$, not shown). Using a lower bound of $4\%$, the results are in line with previous estimates of the victims of the Peruvian internal conflict: $\widehat{n} = 69,970$ ($95\%$ CI: $[64,528, 75,413]$) \citep{patrick2003many, manrique2019estimating, das2023doubly}. This behavior highlights the crucial role of the positivity assumption of the nuisance functions $q_y$'s on the estimated population size, which is often overlooked in empirical applications of capture-recapture methods.

We demonstrate the advantages of our methodology by employing a conditional log-linear model on the DP and CVR lists, conditional on not being captured by the remaining five ODH lists. In this setting, the no highest-order interaction identifying assumption is equivalent to assuming conditional independence between these two lists and therefore we use the estimator \eqref{eq:l2_condind} from Corollary \ref{cor:l2_condind}. The identification assumption that the DP and CVR lists are conditionally independent has been previously used in the literature to estimate the number of casualties given their scope and capture strategies \citep{das2023doubly}. Moreover, we relax this assumption by conditioning on being unobserved by the other five lists from the ODH organizations. The results, show in Figure \ref{fig:peru_2list}, demonstrate that this approach is less sensitivity to the positivity bound giving informative results for $\epsilon$ as small as $1\%$: $\widehat{n} = 67,177$ ($95\%$ CI: $[50,287, 84,067]$). Notably, the point estimates quickly concentrates around approximately $68,000$ victims in line with the previous literature on the Peruvian internal armed conflict. In contrast, the plug-in estimator in this case returns $\widehat{n} = 85,350$ ($95\%$ CI: $[68,446, 102,253]$) for $\epsilon = 0.01$.

\begin{figure}[ht]
    \centering
    \begin{subfigure}{0.48\textwidth}
        \includegraphics[width=\textwidth]{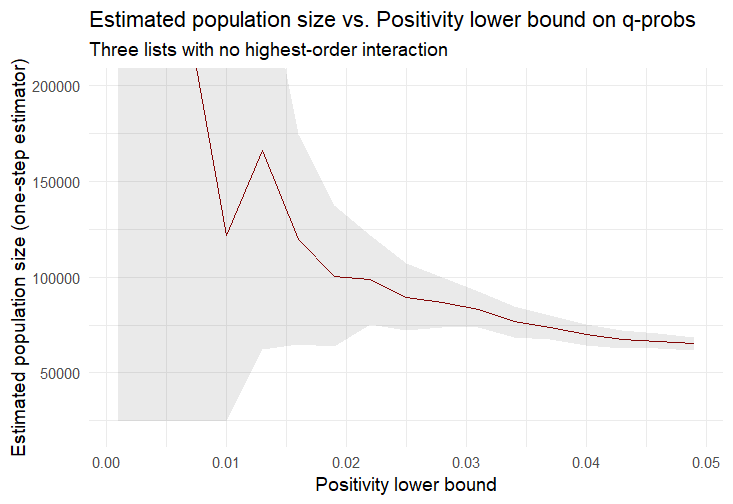}
        \caption{Three lists (CP, CVR, ODH) with no highest-order interaction conditional on covariates.}
        \label{fig:peru_pos}
    \end{subfigure}
    \begin{subfigure}{0.48\textwidth}
        \includegraphics[width=\textwidth]{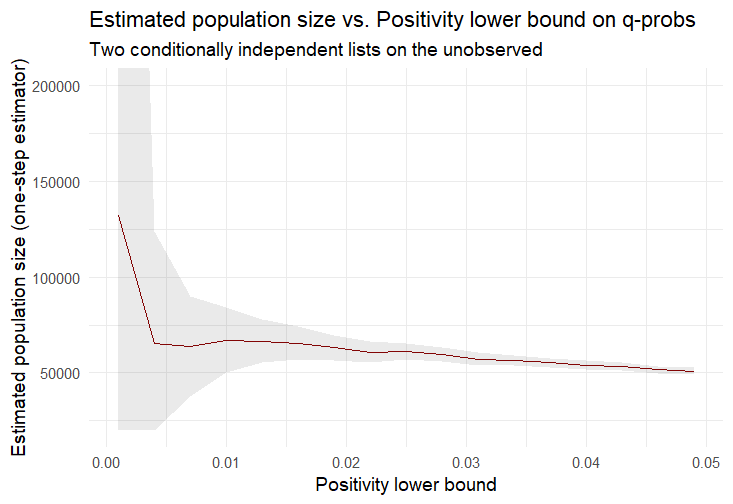}
        \caption{Two independent lists (CVR, DP) conditional on being unobserved by ODH and covariates.}
    \label{fig:peru_2list}
    \end{subfigure}
        
    \caption{Estimated number of casualties in the Peruvian internal conflict for different identification strategies vs. positivity lower bound on the estimated $q$-probabilities.}
\end{figure}

Next, we use the proposed sensitivity analysis framework to estimate bounds that partially identify the population size and analyze how sensitive the results are to deviations from the no highest-order interaction assumption. We present the estimated lower and upper bounds, and the resulting interval for $n$, when we vary the bound $\delta$ on the highest-order interaction term in Figure \ref{fig:sensitivity_analysis}. As expected, setting $\delta = 0$ recovers the no highest-order interaction case, while larger values of $\delta$ yield wider bounds, as they allow for more complex dependence structures between the different capture episodes making identification more difficult. Figure \ref{fig:peru_partial} depicts the sensitivity analysis for the conditional log-linear model with three lists. For instance, using $\delta = 0.2$ the resulting estimated bounds are $\widehat{n}_{\ell} = 63,991$ ($95\%$ CI: $[58,593, 69,389]$) and $\widehat{n}_{u} = 74,121$ ($95\%$ CI: $[68,324, 79,918]$). That is, using a conditional log-linear model with highest-order interaction bounded by $\delta = 0.2$, we estimate the population size to be between $58,593$ and $79,918$. On the other hand, Figure \ref{fig:peru_2lists_partial} shows that the conditional log-linear model among the unobserved using two lists (CVR and DP) is less sensitive to the no highest-order interaction assumption, making the estimated population size more robust to deviations from the conditional independence assumption. With $\delta = 0.2$, the number of casualties is estimated between $53,027$ and $72,622$, with the lower and upper bounds $\widehat{n}_{\ell} = 58,950$ ($95\%$ CI: $[53,027, 64,872]$) and $\widehat{n}_{u} = 66,348$ ($95\%$ CI: $[60,074, 72,622]$). Note that for this analysis we use fixed positivity bounds $\epsilon = 4\%$ and $\epsilon = 2\%$ for the conditional log-linear models with three and two lists, respectively, given our previous results in the no highest-order interaction setting.

\begin{figure}[ht]
    \centering
    \begin{subfigure}{0.48\textwidth}
        \includegraphics[width=\textwidth]{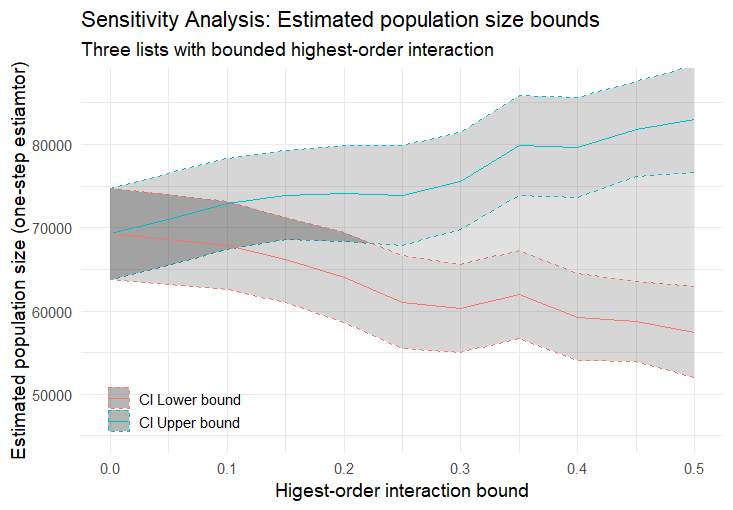}
        \caption{Log-linear model with three lists (CVR, DP and ODH) conditional on covariates ($\epsilon = 4\%$).}
        \label{fig:peru_partial}
    \end{subfigure}
    \begin{subfigure}{0.48\textwidth}
        \includegraphics[width=\textwidth]{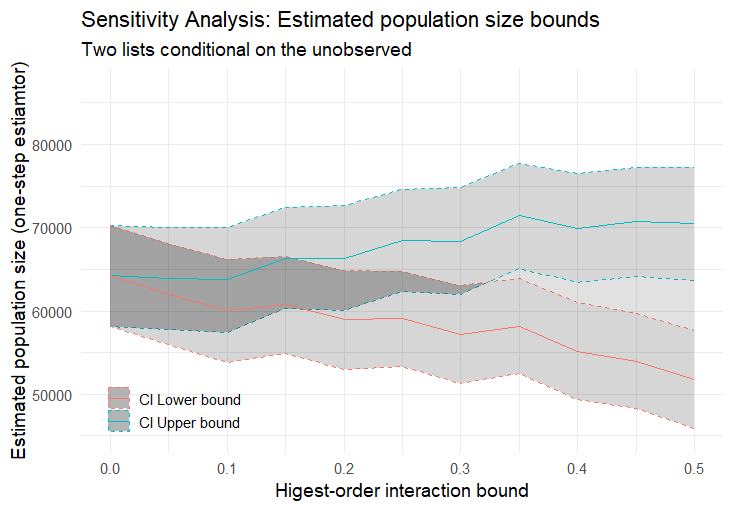}
        \caption{Two lists (CVR and DP) conditional on being unobserved by ODH and covariates ($\epsilon = 2\%$).}
        \label{fig:peru_2lists_partial}
    \end{subfigure}
    
    \caption{Estimated bounds for the number of casualties in the Peruvian internal conflict for different values of the bound $\delta$ on the highest-order interaction in a conditional log-linear model.}
    \label{fig:sensitivity_analysis}
\end{figure}

In summary, our empirical analysis of real data from the Peruvian internal armed conflict demonstrates the key role of the identification assumptions in causality estimation. In traditional log-linear models using smaller values for the positivity assumption makes the estimation more difficult (as theoretically predicted by the efficiency bound in Theorem \ref{thm:eff_bound}), resulting in very wide confidence intervals. However, imposing a large positivity bound may bias the estimation in complicated ways. On the other hand, using larger values for the highest-order interaction bound accounts for more complex dependence structures between the different lists, but makes identification more difficult and yields wider bounds. Crucially, our framework allows us to study how these bounds affect the estimated population size and to incorporate domain knowledge of plausible values. Using the proposed identification strategy on the two largest lists in the Peruvian internal armed conflict DP and CVR, conditional on not being captured by the remaining five ODH lists, we estimate approximately $68,000$ victims during the Peruvian internal armed conflict in line with previous literature.  
\section{Conclusions}

In this paper, we introduced a general and flexible framework for estimating the total size of a closed population from capture-recapture data, addressing key challenges such as heterogeneous capture probabilities, low overlap between lists, and complex dependencies across capture-recapture episodes. We derived a one-step estimator for the capture probability within a novel conditional log-linear model on the unobserved, accommodating heterogeneous and dependent capture probabilities. This approach extends standard assumptions of no highest-order interaction and conditional independence, commonly used in previous literature. Our analysis demonstrates that our estimators are nearly optimal, approximating a sample average of efficient influence function values with high probability when nuisance errors are small in a doubly-robust fashion. Additionally, our estimators are shown to be approximately normal in a Berry-Esseen sense which yield confidence intervals with non-asymptotic guarantees. These properties hold for any number of observed units, even when the $q$-probabilities are flexibly estimated using machine learning models in a data-driven way without compromising statistical guarantees. Furthermore, we extended our results to estimate and perform inference on the total population size $n$. 

Notably, our identification strategy is weaker than assuming a log-linear model for all $K$ lists with no highest-order interaction, requiring only that $J \leq K$ lists are not arbitrarily dependent within the unobserved population, conditional on covariates. 
This approach is both interpretable and actionable: to make the identification of the capture probability easier, an analyst can measure more relevant covariates or collect additional lists. Moreover, it only requires positivity for $2^{J}-1$ of the $q$-probabilities, which is substantially fewer than the $2^{K}-1$ possible capture profiles. Finally, our sensitivity analysis framework allows easy incorporation of domain knowledge about identification assumption parameters and analysis of their effects on the estimated population size. 

For instance, in our analysis of simulated and real data from the Peruvian Truth Commission, we highlighted the critical role of the identifying approach used in population size estimation, emphasizing the need for a more flexible and rigorous methodology. Standard log-linear models were found to be highly sensitive to the identification assumptions: smaller values for the positivity parameter made the estimation more difficult, resulting in very wide confidence intervals, while larger values for the highest-order interaction bound made identification more difficult and yielded wider bounds. In contrast, our proposed identification strategy, using a conditional log-linear model on the two largest lists ($J=2$) conditioned on not being captured by the remaining five ODH lists and covariates, provided a stable and transparent estimation of the victim population, consistent with previous literature.

Our methodology represents a significant contribution to the population size estimation literature, offering a flexible, transparent, and rigorous approach. The proposed methodology addresses recent critiques of capture-recapture models, offering a weaker and more interpretable identifying assumption and accommodating complex heterogeneous capture probabilities depending on high-dimensional or continuous covariates.  We believe that our framework has broad applicability in various domains where understanding the size of a hidden population is critical, such as biology, public health, human rights, and many others.

\bibliographystyle{apalike}
\bibliography{biblio.bib}

\newpage
\section*{Appendix}

\subsection*{A. No highest-order interaction in conditional log-linear models}

In the following, we build some intuition on the interpretation of the highest-order interaction term in (conditional) log-linear models and the corresponding identification assumption. We anchor its interpretation to deviations from a conditional independence assumption and show how it can account for more general capture-recapture settings. Without loss of generality, we state the different identification assumptions in terms of the first two lists $Y_1, Y_2$.

First, note that in the $K=2$ list case the highest-order interaction term can be expressed as
\begin{equation}
    \exp(\alpha_1(X)) = \frac{p_{00}(X)p_{11}(X)}{p_{10}(X)p_{01}(X)}, 
\end{equation}
where $p_{y}(X) = \PP(Y=y \mid X)$ is the true probability of having capture profile $y = (y_1, y_2)$, given covariates $X$. Assuming no highest-order interaction implies
\begin{align*}
    \exp(0) = 1 = \frac{p_{00}(X)p_{11}(X)}{p_{10}(X)p_{01}(X)} & \iff \frac{p_{11}(X)}{p_{01}(X)} = \frac{p_{10}(X)}{p_{00}(X)} \\ & \iff odds(Y_1 =1 \mid Y_2 = 1, X) = odds(Y_1=1 \mid Y_2 = 0, X).
\end{align*}
This condition states that the odds of appearing on list 1 conditional on appearing on list 2 and covariates $X$ is equal to the odds of appearing on list 1 conditional on \textit{not} appearing on list 2 and covariates $X$. Thus, in the $K=2$ list case, the no highest-order interaction assumption is equivalent to conditional independence between lists: $Y_1 \indep Y_2 \mid X$.


Similarly, for the $K=3$ lists case we have that
\begin{align*}
    \exp(\alpha_1(X)) & = \frac{p_{100}(X)p_{010}(X)p_{001}(X)p_{111}(X)}{p_{000}(X)p_{110}(X)p_{101}(X)p_{011}(X)}\\ 
    & = \frac{odds(Y_1 \mid Y_2 = 0, Y_3 = 0, X)odds(Y_1 \mid Y_2 = 1, Y_3 = 1, X)}{odds(Y_1 \mid Y_2 = 1, Y_3 = 0, X)odds(Y_1 \mid Y_2 = 0, Y_3 = 1, X)} \\
    & = \frac{OR_{Y_1, Y_2}(Y_3=1, X)}{OR_{Y_1, Y_2}(Y_3 = 0, X)}.
\end{align*}
To understand this expression, first note that under the conditional independence assumption $Y_1 \indep Y_2 \mid Y_3, X$, we have that $OR_{Y_1, Y_2}(Y_3=1, X) = OR_{Y_1, Y_2}(Y_3 = 0, X) = 1$. Moving away from this assumption, which can be restrictive in several capture-recapture settings, the assumption of no highest-order interaction implies that these odds ratios are equal 
$$OR_{Y_1, Y_2}(Y_3=1, X) = OR_{Y_1, Y_2}(Y_3 = 0, X),$$
but not necessarily equal to 1. 

In the $K=4$ case, the highest-order interaction term in a saturated log-linear model takes the following form, where $odds(Y_1 \mid (y_2y_3y_4), X) = odds(Y_1 \mid Y_2 = y_2, Y_3 = y_3, Y_4 = y_4, X)$,
\begin{align*}
    \exp(\alpha_1(X)) & = \frac{p_{0000}(X)p_{1100}(X)p_{1010}(X)p_{1001}(X)p_{0110}(X)p_{0101}(X)p_{0011}(X)p_{1111}(X)}{p_{1000}(X)p_{0100}(X)p_{0010}(X)p_{0001}(X)p_{1110}(X)p_{1101}(X)p_{1011}(X)p_{0111}(X)}\\ 
    & = \frac{odds(Y_1 | (100), X)odds(Y_1 | (010), X)odds(Y_1 | (001), X)odds(Y_1 | (111), X)}{odds(Y_1 | (000), X)odds(Y_1 | (110), X)odds(Y_1 | (101), X)odds(Y_1 | (011), X)} \\
    & = \frac{OR_{Y_1, Y_2}(Y_3=0, Y_4=0, X)OR_{Y_1, Y_2}(Y_3=1, Y_4=1, X)}{OR_{Y_1, Y_2}(Y_3=1, Y_4=0, X)OR_{Y_1, Y_2}(Y_3 = 0, Y_4=1, X)}.
\end{align*}
Importantly, a similar conditional independence assumption $Y_1 \indep Y_2 \mid Y_3, Y_4, X$ implies that $OR_{Y_1,Y_2}(Y_3=y_3,Y_4=y_4,X) = 1$ for all combinations of $(y_3, y_4)$. The no highest-order interaction assumption relaxes this condition, but still imposes certain symmetry in the ORs such that
$$\frac{OR_{Y_1, Y_2}(Y_3=1, Y_4=1, X)}{OR_{Y_1, Y_2}(Y_3 = 0, Y_4=1, X)} = \frac{OR_{Y_1, Y_2}(Y_3=1, Y_4=0, X)}{OR_{Y_1, Y_2}(Y_3=0, Y_4=0, X)}.$$

Finally, for the general case of $K$ lists, we have that the highest-order interaction satisfies the identity
\begin{equation}
\label{eq:kway-interaction}
    \exp((-1)^{K}\alpha_1(X)) = \frac{\prod_{\{y ~even\}
    }p_{y}(X)}{\prod_{\{y ~odd\}}p_y(X)} = \frac{\prod_{\{(y_3,...,y_K) ~even\}}OR_{Y_1, Y_2}((y_3,...,y_K), X)}{\prod_{\{(y_3,...,y_K) ~odd\}}OR_{Y_1,Y_2}((y_3,...,y_K), X)},
\end{equation}
where ``$y$ even/odd'' denotes, with a slight abuse of notation, that the number of non-zero entries of the capture profile vector $y$ ($L_1$ norm) is even/odd. Analogous to the previous cases, imposing the conditional independence assumption of $Y_1 \indep Y_2 \mid Y_3, ..., Y_K, X$, implies that all the odds ratios appearing in \eqref{eq:kway-interaction} are equal to $1$, which can be too restrictive in real-world applications for complex capture-recapture settings with stronger dependencies across capture episodes. On the other hand, the no highest-order interaction approach can be seen as a relaxation of the conditional independence condition assuming certain symmetry on the ORs such that
$$\prod_{\{(y_3,...,y_K) ~even\}}OR_{Y_1, Y_2}((y_3,...,y_k), X) = \prod_{\{(y_3,...,y_k) ~odd\}}OR_{Y_1,Y_2}((y_3,...,y_k), X).$$
This condition will be satisfied if all ORs are equal to one (i.e., if the conditional independence assumption holds) or if deviations from conditional independence ``cancel out'' across terms. This can be the case, for example, if all but two ORs are equal to $1$ and $OR_{Y_1, Y_2}((1,y_4,...,y_k), X) = OR_{Y_1, Y_2}((0,y_4,...,y_k), X) = 2$. The proposed sensitivity analysis using a bounded highest-order interaction assumption further relaxes the conditional independence assumption and this idealized symmetry in the odds ratios, and allows for these to vary up to the bound $\exp(\delta)$. This accounts for more arbitrary dependencies between capture episodes such that the deviation from this symmetry in ORs is bounded by $\exp(\delta)$.

In general, for any number of $K$ lists in a conditional log-linear model, assuming that two lists are conditionally independent given the other lists and covariates is a stricter identification assumption than the no highest-order  interaction. The latter identifying assumption accounts for more complex capture-recapture dependencies across lists, but is harder to interpret for a number of lists larger than $K>3$. The proposed identification strategy introduced in our work generalizes these two approaches noting that it is sufficient for identification that a \textit{subset} of the $K$ lists has no highest-order interaction, conditional on not being observed by the remaining lists and covariates. For instance, 
$$Y_1 \indep Y_2 \mid X, Y_3 = 0, \dots, Y_K=0,$$
is sufficient to identify the capture probability $\psi$.

\begin{remark}
The conditional independence assumption $Y_1 \indep Y_2 \mid X, Y_3, \dots, Y_K$ further imposes some restrictions on the coefficients of the saturated log linear model implied by the condition that $OR_{Y_1, Y_2}((y_3,...,y_k), X) = 1$ for all $(y_3,...,y_k)$. These extra modeling conditions arising from the conditional independence identifying assumption (which are testable from the available data) constrain the models that can be used for estimating $\QQ$. Therefore, it restricts analysts to semiparametric modeling of the capture probability to distributions that satisfy these additional conditions. On the other hand, the no highest-order interaction assumptions and conditioning only on not being observed by the remaining lists do not impose any restrictions on the observed distribution $\QQ$, allowing for flexible \textit{nonparametric} data-driven modeling of $\psi$. 
\end{remark}

\subsection*{B. Proofs}

\begin{proof}[Proof of Proposition \ref{prop:point_psi}]

The highest-order interaction term in a saturated conditional log-linear model on $J$ lists, conditional on being unobserved by the other $K-J$ lists and covariates, can be expressed as
$$\alpha_1(X) = \sum_{(y_1, \dots, y_J)} (-1)^{J + |y|}\log\left[\PP\{(Y_1, \dots, Y_J)=(y_1, \dots, y_J) \mid X, (Y_{J+1}, \dots, Y_{K}) = 0\}\right].$$
Note that we can write the log conditional probability for each capture profile $(y_1, \dots, y_J)$ as 
$$\log\left[\frac{\PP\{(Y_1, \dots, Y_J, Y_{J+1}, \dots, Y_K)=(y_1, \dots, y_J, 0, \dots, 0)\mid X\}}{\PP\{(Y_{J+1}, \dots, Y_{K}) = 0 \mid X\}}\right],$$
and the terms $\log[\PP\{(Y_{J+1}, \dots, Y_{K}) = 0 \mid X\}]$ will cancel out from the same number of even and odd capture profiles in the sum in $\alpha_{1}(X)$. Then we can use the relationship between the population and the observed distributions $\PP(Y=y \mid X) = q_y(X)\gamma(X)$, for all $y\neq 0$, to rewrite the highest-order interaction coefficient in terms of the $q$-probabilities and the conditional capture probability $\gamma(X)$:
\begin{align*}
    \alpha_1(X) & = \sum_{(y_1, \dots, y_J)} (-1)^{J + |y|}\log[\PP(Y=(y_1, \dots, y_J, 0, \dots, 0) \mid X)] \\
    & = \sum_{(y_1, \dots, y_J) \neq 0} (-1)^{J + |y|}\log[q_{(y_1, \dots, y_J, 0, \dots, 0)}(X)\gamma(X)] + (-1)^{J}\log(1-\gamma(X)). 
\end{align*}
Assuming no highest-order interaction $\alpha_{1}(X) = 0$, we solve for $\gamma(X)$ and use that $\psi^{-1}$ can be expressed as the empirical harmonic mean $\E_{\QQ}[\gamma(X)^{-1}]$ to obtain the identifying expression.

\end{proof}


\begin{proof}[Proof of Lemma \ref{lemma:if_psi}]

We show that the proposed expression 
$$\phi(X,Y) = \left(\frac{1}{\gamma(X)} - 1\right)\sum_{(y_1, \dots, y_J) \neq 0} (-1)^{|y|+1}\frac{\mathbbm{1}(Y = y)}{q_y(X)}  +1 - \frac{1}{\psi}$$
is the Efficient Influence Function for $\psi^{-1}$ by computing the remainder term of its von Mises expansion \citep{bickel1993efficient}. Let $\overbar{\psi} = \psi(\overbar{\QQ})$ be the plug-in estimator of the capture probability for a generic distribution $\overbar{\QQ}$. Then, the remainder of the von Mises expansion is given by 
\begin{align*}
    & R_2(\QQ, \overbar{\QQ}) = \frac{1}{\overbar{\psi}} - \frac{1}{\psi}
    + \E_{\QQ}\left(\phi(X,Y; \overbar{\QQ})\right) \\
    & = \frac{1}{\overbar{\psi}} - \frac{1}{\psi} + \E_{\QQ}\left[\left(\frac{1}{\overbar{\gamma}(X)} - 1\right)\sum_{(y_1, \dots, y_J) \neq 0} (-1)^{|y|+1} \frac{\mathbbm{1}(Y=y)}{\overbar{q}_y(X)} + 1 - \frac{1}{\overbar{\psi}}\right] \\
    & = \E_{\QQ}\left[\left(\frac{1}{\overbar{\gamma}(X)} - 1\right)\sum_{(y_1, \dots, y_J) \neq 0} (-1)^{|y|+1} \frac{q_y(X)}{\overbar{q}_y(X)} + 1 - \frac{1}{\gamma(X)}\right] \\
    & = \E_{\QQ}\left[\left(\frac{1}{\overbar{\gamma}(X)} - 1\right)\sum_{(y_1, \dots, y_J) \neq 0} (-1)^{|y|+1} \left(\frac{q_y(X)}{\overbar{q}_y(X)}-1\right) + \left(\frac{1}{\overbar{\gamma}(X)} - 1\right)\sum_{(y_1, \dots, y_J) \neq 0} (-1)^{|y|+1} + 1 - \frac{1}{\gamma(X)}\right] \\
    & = \E_{\QQ}\left[\left(\frac{1}{\overbar{\gamma}(X)} - 1\right)\sum_{(y_1, \dots, y_J) \neq 0} (-1)^{|y|+1} \left(\frac{q_y(X)}{\overbar{q}_y(X)}-1\right) + \left(\frac{1}{\overbar{\gamma}(X)} - 1\right) - \left(\frac{1}{\gamma(X)} - 1\right)\right].
\end{align*}

The result follows from computing the second-order Taylor expansion for
$$\frac{1}{\gamma(X)} - 1 = \exp\left(\sum_{(y_1,\dots, y_J)\neq 0} (-1)^{1+|y|}\log(q_y(X))\right)$$ 
around $\{\overbar{q}_y(X)\}_{y\neq 0}$, such that
\begin{align*}
    \frac{1}{\gamma} - 1 & = \left(\frac{1}{\overbar{\gamma}} - 1\right) + \left(\frac{1}{\overbar{\gamma}} - 1\right)\sum_{(y_1, \dots, y_J) \neq 0} (-1)^{|y|+1} \left(\frac{q_y}{\overbar{q}_y}-1\right)  \\
    & ~~~ + \left(\frac{1}{\widetilde{\gamma}}-1\right)\left[\sum_{y_i\neq y_j \neq 0} (-1)^{|y_i| + |y_j|}\frac{(q_{y_i} - \overbar{q}_{y_i})}{\widetilde{q}_{y_i}}\frac{(q_{y_j} - \overbar{q}_{y_j})}{\widetilde{q}_{y_j}} + \sum_{|y|\neq 0, ~even} \left(\frac{q_{y} - \overbar{q}_{y}}{\widetilde{q}_{y}}\right)^2\right]
\end{align*}
for some value $\widetilde{q}$ of the $q$-probabilities between $q_{y}$ and $\overbar{q}_y$, where we drop the $X$ argument from $q_y(X)$ and $\gamma(X)$ for ease of exposition. In the last expression, the sums are taken with respect to all possible capture profiles of the form $y = (y_1,\dots, y_J, 0, \dots, 0)$. 

Note that the first two terms in $R_2(\QQ, \overbar{\QQ})$ cancel out with the linear terms in the Taylor expansion for $\gamma^{-1} - 1$.  Then, the remainder is given by
$$R_2(\QQ, \overbar{\QQ}) = \E_{\QQ}\left[\left(\frac{1}{\widetilde{\gamma}}-1\right)\left\{\sum_{y_i\neq y_j \neq 0} (-1)^{|y_i| + |y_j|}\frac{(q_{y_i} - \overbar{q}_{y_i})}{\widetilde{q}_{y_i}}\frac{(q_{y_j} - \overbar{q}_{y_j})}{\widetilde{q}_{y_j}} + \sum_{|y|\neq 0, ~even} \left(\frac{q_{y} - \overbar{q}_{y}}{\widetilde{q}_{y}}\right)^2\right\}\right]$$

The last equality shows that the remainder term $R_2(\QQ, \overbar{\QQ})$ only depends on products of errors of the nuisance functions, which implies in particular that $\frac{d}{d\epsilon}R_{2}(\QQ, \QQ_{\epsilon})\big|_{\epsilon=0} = 0$ for any parametric submodel $\QQ_{\epsilon}$. Therefore, by Lemma 2 in \cite{kennedy2023semiparametric}, this implies that the proposed $\phi(X,Y)$ is the Influence Function for $\psi^{-1}$. Moreover, the influence function is unique and efficient provided that our nonparametric model does not restrict the tangent space \citep{semiparametric, bickel1993efficient}, concluding our proof.

\end{proof}

\begin{proof}[Proof of Theorem \ref{thm:eff_bound}]
The efficiency bound for $\psi^{-1} = \E[\gamma(X)^{-1}]$ can be computed as the variance of its efficient influence function $\phi(X,Y)$. Using the law of total variance, we have that 
$$\Var(\phi(X,Y)) = \Var\left(\left(\frac{1}{\gamma(X)} - 1\right)\sum_{(y_1, \dots, y_J) \neq 0} (-1)^{|y|+1}\frac{\mathbbm{1}(Y = y)}{q_y(X)} + 1 - \frac{1}{\psi} \right) = T_1 + T_2,$$
where
\begin{align*} 
    T_1 & = \Var\left(\E\left[\left(\frac{1}{\gamma(X)} - 1\right)\sum_{(y_1, \dots, y_J) \neq 0} (-1)^{|y|+1}\frac{\mathbbm{1}(Y = y)}{q_y(X)} ~\Bigg|~ X\right]\right), \\
    T_2 & = \E\left[\Var\left(\left(\frac{1}{\gamma(X)} - 1\right)\sum_{(y_1, \dots, y_J) \neq 0} (-1)^{|y|+1}\frac{\mathbbm{1}(Y = y)}{q_y(X)} ~\Bigg|~ X\right)\right].
\end{align*}

The first term can be simplified to 
\begin{align*}
    T_1 & = \Var\left(\left(\frac{1}{\gamma(X)} - 1\right)\sum_{(y_1, \dots, y_J) \neq 0} (-1)^{|y|+1}\frac{\E[\mathbbm{1}(Y = y) \mid X]}{q_y(X)} \right) \\
    & = \Var\left(\left(\frac{1}{\gamma(X)} - 1\right)\sum_{(y_1, \dots, y_J) \neq 0} (-1)^{|y|+1}\right) = \Var\left(\frac{1}{\gamma(X)}\right).
\end{align*}
We now analyze the second term $T_2$:
\normalsize
\begin{align*}
    & T_2 = \E\left[\left(\frac{1}{\gamma(X)} - 1\right)^2\Var\left(\sum_{(y_1, \dots, y_J) \neq 0} (-1)^{|y|+1}\frac{\mathbbm{1}(Y = y)}{q_y(X)} ~\Bigg|~ X\right)\right] \\
    & = \E\left[\left(\frac{1}{\gamma(X)} - 1\right)^2 \left\{\sum_{(y_1, \dots, y_J) \neq 0}\frac{\Var(\mathbbm{1}(Y = y)|X)}{q_y(X)^2} + 2\sum_{y_i \neq y_j \neq 0} \frac{(-1)^{|y_i|+ |y_j|}}{q_{y_i}(X)q_{y_j}(X)} Cov(\mathbbm{1}(Y=y_i), \mathbbm{1}(Y=y_j) | X)  \right\}\right] \\
    & = \E\left[\left(\frac{1}{\gamma(X)} - 1\right)^2 \left\{\sum_{(y_1, \dots, y_J) \neq 0}\frac{q_y(X) (1-q_y(X))}{q_y(X)^2} + 2\sum_{y_i \neq y_j \neq 0} \frac{(-1)^{|y_i|+ |y_j|}}{q_{y_i}(X)q_{y_j}(X)} \left(\E[\mathbbm{1}_{y_i}\mathbbm{1}_{y_j}|X] - q_{y_i}(X)q_{y_j}(X)\right)  \right\}\right] \\
    & = \E\left[\left(\frac{1}{\gamma(X)} - 1\right)^2 \left\{\sum_{(y_1, \dots, y_J) \neq 0}\frac{1-q_y(X)}{q_y(X)} - 2\sum_{y_i \neq y_j \neq 0} \frac{(-1)^{|y_i|+ |y_j|}}{q_{y_i}(X)q_{y_j}(X)}q_{y_i}(X)q_{y_j}(X)  \right\}\right] \\
    & = \E\left[\left(\frac{1}{\gamma(X)} - 1\right)^2 \left\{\sum_{(y_1, \dots, y_J) \neq 0}\frac{1}{q_y(X)} - \sum_{(y_1, \dots, y_J) \neq 0} 1 - 2\sum_{y_i \neq y_j \neq 0} (-1)^{|y_i|+ |y_j|}  \right\}\right] \\
    & = \E\left[\left(\frac{1}{\gamma(X)} - 1\right)^2 \left\{\sum_{(y_1, \dots, y_J) \neq 0}\frac{1}{q_y(X)} - 1 \right\}\right], 
\end{align*}
\normalsize
where the last equality follows from
\begin{align*}
  \sum_{(y_1, \dots, y_J) \neq 0} (-1)^{2|y|} + 2\left(\sum_{y_i\neq y_j\neq 0} (-1)^{|y_i| + |y_j|}\right) = \left(\sum_{(y_1, \dots, y_J) \neq 0} (-1)^{|y|}\right)^2 = 1.
\end{align*}
Putting the two terms together yields the result.
\end{proof}

\begin{proof}[Proof of Theorem \ref{thm:optimal}]

First note that by the definition of our one-step estimator and the von Mises expansion of $\psi$, we have that
\begin{align*}
    \frac{1}{\widehat{\psi}_{os}} - \frac{1}{\psi} - \QQ_{N}(\phi) & = \frac{1}{\widehat{\psi}_{pi}} + \QQ_{N}(\widehat{\phi}) - \frac{1}{\psi} - \QQ_{N}(\phi) \\
    & = -\QQ(\widehat{\phi}) + \widehat{R}_2 + \QQ_{N}(\widehat{\phi}) - \QQ_{N}(\phi) \\
    & = \widehat{R}_2 + (\QQ_{N} - \QQ)(\widehat{\phi} - \phi),
\end{align*}
where $\widehat{R}_2$ is the remainder term of the von Mises expansion \eqref{eq:r2_bound} derived in the proof of Lemma \ref{lemma:if_psi}. The last equality uses the property that the EIF has mean zero, $\QQ(\phi) = 0$ 

Following the proof of Theorem 2 in \cite{das2023doubly} we have that
$$\E\left[\Bigg|\frac{1}{\widehat{\psi}_{os}} - \frac{1}{\psi} - \QQ_{N}(\phi)\Bigg|^2\right] = \E\left[\Big|\widehat{R}_2 + (\QQ_{N} - \QQ)(\widehat{\phi} - \phi))\Big|^2\right] \leq \E[\widehat{R}_2^2] + \frac{1}{N}\E\left[\norm{\widehat{\phi} - \phi}^2\right].$$
The result follows by applying Markov's inequality for error tolerance $\eta > 0$
\begin{align*}
    \PP\left[\Bigg|\frac{1}{\widehat{\psi}_{os}} - \frac{1}{\psi} - \QQ_{N}(\phi)\Bigg| \leq \eta \right] & \geq 1- \frac{1}{\eta^2}\E\left[\Bigg|\frac{1}{\widehat{\psi}_{os}} - \frac{1}{\psi} - \QQ_{N}(\phi)\Bigg|^2\right] \\
    & \geq 1 - \frac{1}{\eta^2}\E\left[\widehat{R}_2^2 + \frac{1}{N}\norm{\widehat{\phi} - \phi}^2\right].
\end{align*}

\end{proof}


\begin{proof}[Proof of Theorem \ref{thm:normal}]

The proof follows the same logic as the proof of Theorem 3 in \cite{das2023doubly}. In detail, note that by the definition of the proposed one-step estimator we have that $\widehat{\psi}^{-1}_{os} - \widehat{\psi} -\widehat{R}_2 = (\QQ_N - \QQ)\widehat{\phi}$ is a sample average of a fixed function conditional on the training sample $Z^{N} = \{(X_i, Y_i\}_{i=1}^{N}$. Therefore, by Berry-Essen Theorem we have that
$$\Phi(t') - \frac{C\rho}{\widetilde{\sigma}^3\sqrt{N}} \leq \PP\left(\frac{\widehat{\psi}^{-1}_{os} - \widehat{\psi} -\widehat{R}_2}{\widetilde{\sigma}/\sqrt{N}} \leq t' ~\Bigg|~ Z^N\right) \leq \Phi(t') + \frac{C\rho}{\widetilde{\sigma}^3\sqrt{N}},$$
for any $t'$ and $N$, where $\widetilde{\sigma}^2 = \Var(\widehat{\phi}\mid Z^N)$ and $\rho = \E[|\widehat{\phi} - \QQ\widehat{\phi}|^3 \mid Z^N]$. 

Taking $t' = \frac{\widehat{\sigma}}{\widetilde{\sigma}}t - \frac{\widehat{R}_2}{\widetilde{\sigma}/\sqrt{N}}$, where $\widehat{\sigma}^2 = \widehat{\Var}(\widehat{\phi})$ is the estimated unconditional variance of the estimated EIF, we can rewrite the previous bounds in terms of the estimation error $\widehat{\psi}^{-1}_{os} - \widehat{\psi}$:
$$\Phi\left(\frac{\widehat{\sigma}}{\widetilde{\sigma}}t - \frac{\widehat{R}_2}{\widetilde{\sigma}/\sqrt{N}}\right) - \frac{C\rho}{\widetilde{\sigma}^3\sqrt{N}} \leq \PP\left(\frac{\widehat{\psi}^{-1}_{os} - \widehat{\psi}}{\widehat{\sigma}/\sqrt{N}} \leq t ~\Bigg|~ Z^N\right) \leq \Phi\left(\frac{\widehat{\sigma}}{\widetilde{\sigma}}t - \frac{\widehat{R}_2}{\widetilde{\sigma}/\sqrt{N}}\right) + \frac{C\rho}{\widetilde{\sigma}^3\sqrt{N}}.$$
Finally, note that using the mean value theorem we can bound the difference in cumulative distributions by
\begin{align*}
    \Bigg|\Phi\left(\frac{\widehat{\sigma}}{\widetilde{\sigma}}t - \frac{\widehat{R}_2}{\widetilde{\sigma}/\sqrt{N}}\right) - \Phi(t)\Bigg| = \Bigg|\Phi'(t_N)\left[\frac{\widehat{\sigma}}{\widetilde{\sigma}}t - \frac{\widehat{R}_2}{\widetilde{\sigma}/\sqrt{N}} - t\right]\Bigg| \leq \frac{1}{\sqrt{2\pi}}\left(\Big|\frac{\widehat{\sigma}}{\widetilde{\sigma}}- 1\Big||t| +  \frac{|\widehat{R}_2|}{\widetilde{\sigma}/\sqrt{N}}\right),
\end{align*}
for some $t_N \in (t, \frac{\widehat{\sigma}}{\widetilde{\sigma}}t - \frac{\widehat{R}_2}{\widetilde{\sigma}/\sqrt{N}})$, where the last inequality uses that $\text{sup}_t \Phi'(t) \leq 1/\sqrt{2\pi}$ and the triangle inequality. The result follows by iterated expectations:
$$\Bigg| \mathbb{P}\left(\frac{\widehat{\psi}^{-1}_{os} - \psi^{-1}}{\widehat{\sigma}/\sqrt{N}} \leq t \right) - \Phi(t)\Bigg| \leq \frac{1}{\sqrt{2\pi}}\left(|t|\E_{\mathbb{Q}}\left[\Big|\frac{\widehat{\sigma}}{\widetilde{\sigma}}-1\Big|\right]  + \sqrt{N}\E_{\mathbb{Q}}\left[\frac{|\widehat{R}_2|}{\widetilde{\sigma}}\right]\right) + \frac{C}{\sqrt{N}}\E_{\mathbb{Q}}\left[\frac{\rho}{\widetilde{\sigma}^3}\right].$$

\end{proof}

\begin{proof}[Proof of Proposition \ref{prop:partial_psi}]
    
Following the proof of Proposition \ref{prop:point_psi}, we can express the highest-order coefficient of the proposed log-linear model as  
    $$\alpha_1(X) = \sum_{(y_1, \dots, y_J) \neq 0} (-1)^{J + |y|}\log(q_{y}(X)\gamma(X)) + (-1)^{J}\log(1-\gamma(X)),$$
for $q$-probabilities $q_y(X) = \QQ(Y = (y_1, \dots, y_J, 0, \dots, 0) \mid X)$. Then, a bounded highest-order coefficient $|\alpha_1(X)| \leq \delta$ implies that the conditional capture probability is partially identified by
\begin{equation*}
\label{eq:part_gamma}
    \frac{1}{\gamma(X)} \in \left[1 + \exp\left(\sum_{(y_1, \dots, y_J) \neq 0} (-1)^{|y| + 1}\log(q_y(X)) \pm \delta \right) \right].
\end{equation*}

Next, note that the parameter $\delta$ must be set such that $\gamma(X) \in [\epsilon, 1]$ to ensure the validity of the conditional capture probabilities it yields. Consequently, we derive the following upper and lower bounds on $\gamma(X)^{-1}$ to partially identify our target parameter while incorporating such constraint on $\delta$:
\begin{equation}
\label{eq:part_gamma}
    \frac{1}{\gamma(X)} \in 
    \left[\frac{1}{\gamma_{\ell}(X)}, \frac{1}{\gamma_{u}(X)}\right] = \left[\min \left\{ 1 + \exp\left(\sum_{(y_1, \dots, y_J) \neq 0} (-1)^{|y| + 1}\log(q_y(X)) \pm \delta \right), \frac{1}{\epsilon} \right\} \right].
\end{equation}
Finally, we establish bounds on the inverse capture probability $\psi^{-1} \in [\psi_{\ell}^{-1}, \psi_{u}^{-1}]$, defined by $\psi_{\ell}^{-1} = \E_{\QQ}\left[\gamma_{\ell}^{-1}(X)\right]$ and $\psi_{u}^{-1} = \E_{\QQ}\left[\gamma_{u}^{-1}(X)\right]$, derived from the lower and upper bounds in \eqref{eq:part_gamma}, respectively.

\end{proof}

\begin{proof}[Proof of Theorem \ref{thm:optimal_bounds}]

We follow the same approach as Lemma \ref{lemma:if_psi} to derive the remainder term in the von Mises expansion for $\psi^{-1}_{\delta}$. We omit the arguments of the functions for clarity. Then, We have that for a generic distribution $\overbar{\QQ}$,
\begin{align*}
    & R_{2,\delta}(\QQ, \overbar{\QQ}) = \frac{1}{\overbar{\psi}_{\delta}} - \frac{1}{\psi_{\delta}}
    + \E_{\QQ}\left(\overbar{\phi}_{\delta}\right) \\
   & = \frac{1}{\overbar{\psi}_{\delta}} - \frac{1}{\psi_{\delta}} + \E_{\QQ}\left[\left(\overbar{\varphi}_{\delta}-\frac{1}{\epsilon}\right)\mathbbm{1}\left(\frac{1}{\overbar{\gamma}_\delta} - \frac{1}{\epsilon} \leq 0 \right)  + \frac{1}{\epsilon} - \frac{1}{\overbar{\psi}_{\delta}}\right] \\
   & = \E_{\QQ}\left[\left(\overbar{\varphi}_{\delta}-\frac{1}{\epsilon}\right)\mathbbm{1}\left(\frac{1}{\overbar{\gamma}_\delta} - \frac{1}{\epsilon} \leq 0 \right)  + \frac{1}{\epsilon} - \left(\frac{1}{\gamma_{\delta}}-\frac{1}{\epsilon}\right)\mathbbm{1}\left(\frac{1}{\gamma_\delta} - \frac{1}{\epsilon} \leq 0 \right)  - \frac{1}{\epsilon}\right] \\
   & = \E_{\QQ}\left[\left(\overbar{\varphi}_{\delta}-\frac{1}{\gamma_{\delta}}\right)\mathbbm{1}\left(\frac{1}{\overbar{\gamma}_\delta} - \frac{1}{\epsilon} \leq 0 \right) + \left(\frac{1}{\gamma_{\delta}} - \frac{1}{\epsilon}\right)\left[\mathbbm{1}\left(\frac{1}{\overbar{\gamma}_\delta} - \frac{1}{\epsilon}\leq 0\right) -\mathbbm{1}\left(\frac{1}{\gamma_\delta} - \frac{1}{\epsilon} \leq 0 \right) \right]\right].
\end{align*}

The first term can be bounded above using an analogous remainder term from Lemma \ref{lemma:if_psi} given the symmetry of the EIFs for $\gamma$ and $\gamma_{\delta}$:
\begin{align*}
    \E_{\QQ}\left[\left(\overbar{\varphi}_{\delta}-\frac{1}{\gamma_{\delta}}\right)\mathbbm{1}\left(\frac{1}{\overbar{\gamma}_\delta} - \frac{1}{\epsilon} \leq 0 \right)\right] & \leq \E_{\QQ}\left[\Bigg|\overbar{\varphi}_{\delta}-\frac{1}{\gamma_{\delta}}\Bigg|\right] \\
    & \leq \E_{\QQ}\left[\sum_{y_i\neq y_j \neq 0} \norm{q_{y_i} - \overbar{q}_{y_i}} \norm{q_{y_j} - \overbar{q}_{y_j}} + \sum_{|y|\neq 0, ~even} \norm{q_{y} - \overbar{q}_{y}}^2\right],
\end{align*}
where the sums are taken over capture profiles of the form $y = (y_1,\dots,y_J,0,\dots,0)$.

For the second term, we have that
\begin{align*}
    \E_{\QQ}&\left[\left(\frac{1}{\gamma_{\delta}} - \frac{1}{\epsilon}\right)\left[\mathbbm{1}\left(\frac{1}{\overbar{\gamma}_\delta} - \frac{1}{\epsilon}\leq 0\right) -\mathbbm{1}\left(\frac{1}{\gamma_\delta} - \frac{1}{\epsilon} \leq 0 \right) \right]\right] \\
    & \leq \E_{\QQ}\left[\Bigg|\frac{1}{\gamma_{\delta}} - \frac{1}{\epsilon}\Bigg|\Bigg|\mathbbm{1}\left(\frac{1}{\overbar{\gamma}_\delta} - \frac{1}{\epsilon}\leq 0\right) -\mathbbm{1}\left(\frac{1}{\gamma_\delta} - \frac{1}{\epsilon} \leq 0 \right) \Bigg|\right] \\
    & \leq \E_{\QQ}\left[\Bigg|\frac{1}{\gamma_{\delta}} - \frac{1}{\epsilon}\Bigg|\mathbbm{1}\left(\Bigg|\frac{1}{\gamma_\delta} - \frac{1}{\epsilon} \Bigg| \leq \Bigg|\frac{1}{\gamma_\delta} - \frac{1}{\overbar{\gamma}_{\delta}}\Bigg| \right)\right] \\
    & \leq \E_{\QQ}\left[\Bigg|\frac{1}{\gamma_{\delta}} - \frac{1}{\overbar{\gamma}_{\delta}}\Bigg|\mathbbm{1}\left(\Bigg|\frac{1}{\gamma_\delta} - \frac{1}{\epsilon} \Bigg| \leq \Bigg|\frac{1}{\gamma_\delta} - \frac{1}{\overbar{\gamma}_{\delta}}\Bigg| \right)\right] \\
    & \leq \norm{\frac{1}{\gamma_{\delta}} - \frac{1}{\overbar{\gamma}_{\delta}}}_{\infty}\PP\left(\Bigg|\frac{1}{\gamma_\delta} - \frac{1}{\epsilon} \Bigg| \leq \norm{\frac{1}{\gamma_\delta} - \frac{1}{\overbar{\gamma}_{\delta}}}_{\infty} \right)\\
    & \leq C\norm{\frac{1}{\gamma_{\delta}} - \frac{1}{\overbar{\gamma}_{\delta}}}_{\infty}^{1+\beta}
\end{align*}
where the second inequality follows from Lemma 1 in \cite{kennedy2020sharp} and the last inequality by the margin condition \eqref{eq:margin_cond}.

Finally, following similarly to Theorem \ref{thm:optimal}, we have by Markov's inequality that for any error tolerance $\eta > 0$
\begin{align*}
    \PP\left[\Bigg|\frac{1}{\widehat{\psi}_{\delta}} - \frac{1}{\psi_{\delta}} - \QQ_{N}(\phi_{\delta})\Bigg| \leq \eta \right] \geq 1 - \frac{1}{\eta^2}\E\left[\widehat{R}_{2,\delta}^2 + \frac{1}{N}\norm{\widehat{\phi}_{\delta} - \phi_{\delta}}^2\right],
\end{align*} 
concluding our proof.

\end{proof}

\begin{proof}[Proof of Theorem \ref{thm:normal_bounds}]

The proof follows the same logic that the proof of Theorem \ref{thm:normal} noting that the proposed one step estimator takes the form $\widehat{\psi}^{-1}_{\delta} - \widehat{\psi} -\widehat{R}_{2,\delta} = (\QQ_N - \QQ)\widehat{\phi}_{\delta}$, which is a sample average of a fixed function conditional on the training sample, where $\widehat{R}_{2,\delta}$ is the remainder term derived in Theorem \ref{thm:optimal_bounds}.
    
\end{proof}

\begin{proof}[Proof of Theorem \ref{thm:ci_n}]

By construction of our one-step estimator we have that 
\begin{align*}
    \frac{1}{\widehat{\psi}_{os}} - \frac{1}{\psi}  & = \frac{1}{\widehat{\psi}_{pi}} + \QQ_{N}(\widehat{\phi}) - \frac{1}{\psi} = -\QQ(\widehat{\phi}) + \widehat{R}_2 + \QQ_{N}(\widehat{\phi}),
\end{align*}
where $\widehat{R}_2$ is the remainder term of the von Mises expansion for $\psi^{-1}$ \eqref{eq:r2_bound} and $\QQ(\widehat{\phi}) = \int \widehat{\phi}(z) d\QQ(z)$. Therefore, applying Theorem 4 in \cite{das2023doubly} we have that
\begin{equation*}
    \widehat{CI}(n) = \left[ \widehat{n}_{os} \pm z_{\alpha/2}\sqrt{\widehat{n}_{os}\left(\widehat{\psi}_{os}\widehat{\sigma}^2 + \frac{1-\widehat{\psi}_{os}}{\widehat{\psi}_{os}} \right)}\right],
\end{equation*}
is an approximately $(1-\alpha)$ confidence interval for the population size $n$ centered at $\widehat{n}_{os} = N/\widehat{\psi}_{os}$, 
where $\widehat{\sigma}^2 = \widehat{\Var}(\widehat{\phi})$ is the unbiased empirical variance of the estimated influence function.
    
\end{proof}

\begin{proof}[Proof of Theorem \ref{thm:cin_bounds}]

The results follow from Theorem \ref{thm:ci_n} noting that our estimator for the generic bound is of the required form
$$\frac{1}{\widehat{\psi}_{\delta}} - \frac{1}{\psi} = -\QQ(\widehat{\phi}_\delta) + \widehat{R}_{2,\delta} + \QQ_{N}(\widehat{\phi}_\delta),$$
derived in Theorem \ref{thm:optimal_bounds}.

\end{proof}

\end{document}